\documentclass[toc=listof,toc=bibliography]{scrartcl}
\pdfoutput=1 
\usepackage{a4wide}

\usepackage{amsmath}
\usepackage{amsthm}
\usepackage{thmtools}

\theoremstyle{definition}
\newtheorem{definition}{Definition}[section]
\newtheorem{example}[definition]{Example}
\newtheorem{assumption}[definition]{Assumption}
\newtheorem{examples}[definition]{Examples}
\newtheorem{construction}[definition]{Construction}
\newtheorem{problem}[definition]{Open Problem}

\theoremstyle{plain}

\newtheorem{theorem}[definition]{Theorem}%[section]
\newtheorem{lemma}[definition]{Lemma}
\newtheorem{proposition}[definition]{Proposition}
\newtheorem{corollary}[definition]{Corollary}

\theoremstyle{remark}
\newtheorem{remark}[definition]{Remark}

\usepackage[final,inline]{fixme}
\usepackage{etoolbox}
\newbool{full}
\boolfalse{full}

\newcommand{\takeout}[1]{\empty}

\usepackage{rotating}

\usepackage{enumerate} % better enum environments

\usepackage{amssymb}
\usepackage{eufrak}
\usepackage{hyperref}
\usepackage{array}
\usepackage{multirow}
\usepackage{color}
\usepackage{mathpartir} %UNCOMMENTED FOR THE FULL VERSION
\usepackage{mathrsfs}
\usepackage{relsize}
\usepackage{lscape}
\usepackage{%\pgfrcspath 
pgfrcs}
\usepackage{%\pgfpath 
pgf}
\usepackage{%\tikzpath 
tikz}
\usetikzlibrary{%
  shapes.misc,% wg. rounded rectangle
  shapes.geometric,%shapes.arrows,%
  positioning,% wg. " of "
  shadows%
}
\usepackage[all]{xy}
\xyoption{2cell}
\xyoption{curve}
\UseTwocells
\SelectTips{cm}{}
\usepackage{stmaryrd}
\usepackage{pigpen}  %For pullback
\usepackage{dsfont}

\numberwithin{equation}{section}

\newcommand{\modcol}[1]{{\color{blue}#1}}

\newcommand{\ibox}{\modcol{%\rightslice\!\!\cdot}}
\blacktriangleright}}
\newcommand{\unpdg}{\ddagger} %unparametric dagger

\newcommand{\comp}{\mathop{\triangleright}}

\newcommand{\curry}[3]{\mathsf{curry}^{#1}_{#2,#3}}

\newcommand{\eval}[2]{\mathsf{eval}_{#1,#2}}
\newcommand{\iprod}[2]{\mathsf{can}^{-1}_{#1,#2}}
\def\ev{\mathsf{eval}}

\newcommand{\WG}{\mathfrak{W}}

\newcommand{\presh}[2]{\mathsf{presh}(#1,#2)}%{\widehat{#1}^{#2}}

\newcommand{\fxtrm}{\modcol{\mathfrak{f}}}
\newcommand{\dltrm}{\modcol{\mathfrak{p}}\!}

\newcommand{\dfsp}{\;}

\newcommand{\fix}[2]{\fxtrm#1.\!\!\dfsp#2} %Church or Curry?
\newcommand{\delay}[1]{\dltrm\dfsp#1}

\newcommand{\deq}{=}%{\, \widehat{=} \,} %definitional equality

\newcommand{\klei}[1]{{#1}^*}
\newcommand{\iter}[1]{\mathsf{it}(#1)}
\newcommand{\om}{\mathtt{\omega}}
\newcommand{\fxob}{\Omega}
\newcommand{\mS}{T}
\newcommand{\siin}{\mathtt{\sigma}}

\def\mtrlo{\xymatrixcolsep{5pc}\xymatrix@1}
\def\mtr{\xymatrix@1}

\newcommand{\commu}{\circlearrowleft }

\newcommand{\cat}[1]{\mathcal{#1}}

\newcommand{\catC}{\cat{C}}
\newcommand{\catD}{\cat{D}}

\def\A{\cat A}
\def\C{\catC}
\def\D{\catD}
\def\cpo{\mathsf{CPO}}
\def\cpob{\mathsf{CPO}_\bot}
\def\cms{\mathsf{CMS}}

\def\setc{\mathsf{Set}}

\def\refeq#1{(\ref{#1})}
\def\prl{\pi_\ell}
\def\prr{\pi_r}
\def\can{\mathsf{can}}
\def\Id{\mathsf{Id}}
\def\id{\mathsf{id}}
\def\ol#1{\overline{#1}}
\def\sol#1{#1^\dagger}
\def\ssol#1{#1^\unpdg}
\def\Tr{\mathsf{Tr}}
\def\op{\mathsf{op}}

\def\subto{\hookrightarrow}
\def\inr{\mathsf{inr}}
\def\inl{\mathsf{inl}}
\def\can{\mathsf{can}}

\newcommand{\tragger}{{\Tr_\dagger}}
\newcommand{\dace}{{\dagger_\Tr}}

\newcommand{\Nat}{\mathds{N}}

\newcommand{\tlnt}[1]{\tlnote[inline,marginclue]{#1}}

\FXRegisterAuthor{sm}{asm}{SM}%Stefan
\FXRegisterAuthor{tl}{atl}{TML}%Tadeusz

\begin{document}
%\setcounter{page}{1001}
%\issue{invited to a special issue (FiCS'13)}
% Author macros %%%%%%%%%%%%%%%%%%%%%%%%%%%%%%%%%%%%%%%%%%%%%%%%
%\title{Guarded Conway Categories\footnote{This work was partially
%supported by someone \dots or was it?}}
\title{Guard Your Daggers and Traces: Properties of Guarded (Co-)recursion}
\subtitle{invited to a special issue of Fundamenta Informaticae (FiCS'13)}

%\address{

\author{Stefan Milius \and Tadeusz Litak}
\titlehead{Chair for Theoretical Computer Science (Informatik 8) \\
Friedrich-Alexander University Erlangen-N\"{u}rnberg, Germany \\
\texttt{mail@stefan-milius.eu \ \ tadeusz.litak@fau.de}}

\maketitle

\begin{abstract}
  Motivated by the recent interest in models of guarded (co-)recursion,
  we study their equational properties. We formulate axioms for guarded
  fixpoint operators generalizing the axioms of iteration theories of
  Bloom and \'Esik.  Models of these axioms include both standard
  (e.g., cpo-based) models of iteration theories and models of guarded
  recursion such as complete metric spaces or the topos of trees
  studied by Birkedal et al. We show that the standard result on the
  satisfaction of all Conway axioms by a unique dagger operation
  generalizes to the guarded setting. We also introduce the notion of
  guarded trace operator on a category, and we prove that guarded
  trace and guarded fixpoint operators are in one-to-one
  correspondence. Our results are intended as first steps leading, hopefully, towards
  future description of classifying theories for guarded recursion. % and
  %hence completeness theorems. % involving our axioms of guarded fixpoint
  %operators in future work.
\end{abstract}

\section{Introduction}

Our ability to  describe concisely potentially infinite
computations or infinite behaviour of systems relies on 
recursion, corecursion and iteration. Most programming
languages and specification formalisms include a fixpoint
operator. In order
to give semantics to such operators one usually considers either
\begin{itemize}
\item models based on complete partial orders where fixpoint operators are
interpreted by  least fixpoints using the Kleene-Knaster-Tarski
theorem or
\item models based on complete metric spaces and unique
fixpoints via Banach's theorem or 
\item  term models where unique
fixpoints arise by unfolding specifications
syntactically. 
\end{itemize}

In the last of these cases, %of unique fixpoints one has to impose syntactic
%restrictions to ensure productivity. Hence, 
one only considers \emph{guarded} (co-)recursive definitions;  see
e.g.~Milner's solution theorem for CCS~\cite{milner89} or Elgot's
iterative theories~\cite{elgot75}. %or Escardo's discussion in \cite{esc99}. 
Thus, the fixpoint operator becomes a
partial operator defined only on a special class of maps. For a
concrete example, consider complete metric spaces which form a category
with all non-expansive maps as morphisms, but unique fixpoints are
taken only of contractive maps.

Recently, there has been a wave of interest in expressing guardedness by a
new type constructor $\ibox$, a kind of ``later'' modality, which
allows to make the fixpoint operator total
 \cite{Nakano00:lics,Nakano01:tacs,AppelMRV07:popl,BentonT09:tldi,KrishnaswamiB11:lics,KrishnaswamiB11:icfp,BirkedalMSS12:lmcs,BirkedalM13:lics,AtkeyMB13:icfp,Litak14:trends}.  For example, in the case
%, see, e.g.,
% Nakano~\cite{Nakano00:lics,Nakano01:tacs}, Appel et al. \cite{AppelMRV07:popl}, Benton and
%Tabareau~\cite{BentonT09:tldi}, Krishnaswami and
%Benton~\cite{KrishnaswamiB11:lics,KrishnaswamiB11:icfp}, Birkedal et
%al.~\cite{BirkedalMSS12:lmcs,BirkedalM13:lics} and Atkey and
%McBride~\cite{AtkeyMB13:icfp}.  For example, in the case
of complete metric spaces, $\ibox$ can be  an endofunctor scaling the metric
of any given space by a fixed factor $0<r<1$ so that non-expansive
maps of type $\ibox X \to X$ are precisely $r$-contractive ones. This allows to define a guarded (parametrized) fixpoint operator on the model that assigns to \emph{every} morphism $e: \ibox X \times Y
\to X$ a morphism $\sol e: Y \to X$. %Another model for %allowing the interpretation
 Languages with a guarded fixpoint operator can be also interpreted in  the ``topos of
trees'', i.e.,\ presheaves on $\omega^{op}$~\cite{BirkedalMSS12:lmcs}
or, more generally, sheaves on complete Heyting
algebras with a well-founded basis~\cite{DiGianantonioM04:fossacs,BirkedalMSS12:lmcs}. Note that by using $\ibox$, guarded recursion becomes a generalization of standard recursion (since $\ibox$ can be the identity functor) rather than a specialization as in previous approaches.\smnote{I added this sentence to address the referee's ``General comment.''}

This paper initiates the study of the essential %equational
properties of such operators. %guarded fixpoint operators $e \mapsto \sol e$. 
%In the realm of ordinary
%fixpoint operators,   
Iteration theories~\cite{be93} are known to %are well-known to %completely 
axiomatize equalities of \emph{unguarded} fixpoint terms in models based on complete partial
orders (see also \cite{sp00}). We make here the
first steps towards similar completeness results in the guarded setting. 

We begin with formalizing the notion of a guarded fixpoint operator on a cartesian category. We discuss a number of models, including not only all those mentioned above, but also some  not mentioned so far in the context of $\ibox$-guarded (co-)recursion. %guarded by a ``later'' modality. 
In fact, we consider the inclusion of examples such as the lifting
functor on $\cpo$ and, more broadly, 
 \emph{let-ccc's with a fixpoint object}~\cite{cp92} (see Examples
\ref{ex:cats}.(\ref{ex:cpo})--(\ref{ex:letccc}) and Theorem~\ref{th:proplet}) %\tlnt{Btw, aren't we writing numbers of example
                  %clauses sometimes in brackets and sometimes not?
                  %Need to unify it atsome point, guess it is just a
                  %matter of using refeq consistently} 
or completely iterative monads (Section \ref{sec:cim}) a pleasant
by-product of our work and a potentially fruitful connection for
future research.

In Section  \ref{sec:prop}, we formulate $\ibox$-guarded generalizations of standard
axioms of Conway and iteration theories (see, e.g.,~\cite{be93,sp00})  and prove their
soundness. %in all models under consideration. 
In particular, %the central result of Section~\ref{sec:fix} is 
models with \emph{unique} guarded fixpoint
operators satisfy all our axioms (Theorem~\ref{thm:unique}). Without the assumption of uniqueness, some problems appear (notably, Open Problems \ref{prob:letdd}, \ref{prob:dinat} and \ref{prob:donetwo}) and generalizations of several known derivations, like that of the Beki\v{c} identity from the Conway axioms (Proposition \ref{prop:bekic}) require some ingenuity. We believe these are positive signs:  sticking $\ibox$ in ``all the
right places'' cannot always be done on autopilot and subtle aspects of (co)-recursion invisible to the unguarded eye come to light, even on the purely equational level. For natural examples, however, most  properties in question seem to hold even without requiring uniqueness, as witnessed, e.g., by Theorems \ref{th:proplet} and \ref{th:cpoconway}.  
%This alone covers most examples discussed in this paper. 
%\tlnt{mention discussion of dinaturality?}\smnt{no, I think
 % ``additional equational properties below'' is enough}

%. This covers the examples of complete
%metric spaces, the topos of trees and (more generally) presheaves on
%well-founded preorders. %This is our first result.

%Section~\ref{sec:tr} focuses on the equivalence
%of fixpoint operators and traced cartesian structures. Recall that
%Joyal, Street and
%Verity~\cite{jsv96} introduced traced monoidal categories and  
Hasegawa~\cite{h97} proved that giving a parametrized fixpoint operator on a cartesian category  is
equivalent to giving a \emph{traced cartesian structure} \cite{jsv96} on that
category.\footnote{\emph{Cartesian} here refers to the monoidal
  product being the ordinary categorical product.} In Section~\ref{sec:tr}, % lifts this result to the guarded setting. 
  we introduce a natural notion of a \emph{guarded trace operator}
on a category, and we prove in Theorem~\ref{thm:tr} that guarded traces
and guarded fixpoint operators are in one-to-one correspondence.
This extends to an isomorphism between the (2-)categories of guarded traced
cartesian categories and guarded Conway categories (Corollary \ref{cor:iso}). Just like in the unguarded case, the notion of trace would make sense in a general monoidal setting (Remarks \ref{rem:mono} and \ref{rem:comono}). We leave this as an exciting avenue for future research.

Finally, Section~\ref{sec:conc} concludes and discusses further work. 

A few words are due on differences with a previously published extended abstract of this paper~\cite{ml13}.
We obviously provide full proof details of all results. We also discuss
additional equational properties of guarded fixpoint operators in
Sections~\ref{sec:bekic} and~\ref{sec:dinat}. Moreover,
Theorem~\ref{th:proplet} concerning let-ccc's is new and so is, e.g., Example \ref{ex:bizjak} (the last provided by Ale\v{s} Bizjak).
%\tlnt{Again: mention discussion of dinaturality?}\smnt{I think the
 % last but one sentence is enough.}

We decided to move some more technical proofs to an appendix in order to make the paper more readable.
%Proofs of the major theorems %are omitted for space considerations. They 
%will be made available in the full version.
%
%\tlnt{This obviously is supposed to be the full version \dots}

\subsection{Notational Conventions}

We assume  familiarity with basic notions
of category theory. We denote the product of two objects by
$
\xymatrix@1{
  A & A \times B \ar[l]_-{\prl} \ar[r]^-{\prr} & B
  }
$
and $\Delta: A \to A \times A$ denotes the diagonal. For every functor
$F$ we write $\can = \langle F\prl, F\prr\rangle: F(A\times B) \to FA
\times FB$ for the canonical morphism. We denote the terminal object
in a cartesian category as 1 and the unique morphism for each $X$ as
$! : X \to 1$. Wherever convenient, we use freely other standard 
conventions such as identifying $X$ and $1 \times X$ or dropping subscripts of natural transformations if they are clear from the context.

 $\cpo$ denotes the
category of complete partial orders (cpo's), i.e.\ partially ordered sets (possibly without a least element) having joins of $\omega$-chains. The
morphisms of $\cpo$ are Scott-continuous maps, i.e.\ maps preserving
joins of $\omega$-chains.  $\cpob$ is the full subcategory of
$\cpo$ given by all cpo's with a least element $\bot$. We will also
consider the category $\cms$ of complete 1-bounded metric spaces and
non-expansive maps, i.e.\ maps $f: X \to Y$ such that for all $x,y \in
X$, $d_Y(fx,fy) \leq d_X(x,y)$; see Krishnaswami and Benton~\cite{KrishnaswamiB11:lics,KrishnaswamiB11:icfp}
      or Birkedal et al.~\cite[Section~5]{BirkedalMSS12:lmcs} and references therein.

Instead of writing ``the following square commutes'' or ``the
following diagram commutes'', we write $\commu$ in the middle of the
diagram in question. We  also use %the convention of simply writing
objects to denote their identity morphisms. Finally, we sometimes write $X = Y$ to indicate that two objects in a category are isomorphic.

%\tlnt{Assignment sign. Sometimes we use $\deq$, sometimes ordinary
 % equality. No consistency whatsoever. What do we choose?} \smnt{I'd
 % say we do not need a special sign; when we say ``let $x = \cdots$''
 % it is clear we mean assignment. But I am fine with whatever you
 % decide and implement in this matter.}

\section{Guarded Fixpoint Operators}
\label{sec:fix}

In this section we define the notion of a guarded fixpoint operator on
a cartesian category and present an extensive list of examples. Some
of these examples like the lifting functor $(-)_\bot$ on $\cpo$ (see
Example \ref{ex:cats}.\ref{ex:cpo}) or completely iterative monads
(see Section \ref{sec:cim}) do not seem to have been considered as
instances of the guarded setting before. We also discuss in detail the connection with \emph{models of guarded fixpoint terms} of Birkedal et al. \cite{BirkedalMSS12:lmcs}, see Proposition \ref{prop:mgrt}.

%We then introduce in
%Section~\ref{sec:prop} (equational) properties of guarded fixpoint
%operators. These properties are motivated by thier unguarded
%counterparts in iteration theories of Bloom and
%\'Esik~\cite{be93}.  We show in Theorem~\ref{thm:unique} that unique
%fixpoint operators satisfy all the properties in question. We also
%establish in Theorem~\ref{th:proplet} that the guarded fixpoint
%operators in let-ccc's of Crole and Pitts satisfy all but possibly one
%of them, i.e., the guarded version of the \emph{double dagger
 % identity}, which however is satisfied at least by  the leading concrete example of let-ccc's: the lifting monad on
%$\cpo$'s. %, that one property we leave as open problem (however, we
%show that the leading example of let-ccc's, cpo's with the lifting
%monad does satisfy that property).
%Two other guarded counterparts of standard equational properties,
%i.e., \emph{the Beki\v{c} identity} and \emph{dinaturality}, merit 
%separate discussions in Sections~\ref{sec:bekic} and~\ref{sec:dinat}, respectively.
 %Finally, we study in
%Subsections~\ref{sec:bekic} and~\ref{sec:dinat} further equational
%properties of guarded fixpoint operators. 

%\tlnt{First, this is practically copy-and-paste of a paragraph in the
 % appendix one page earlier. Second, it's old: shall we mention, e.g,
 % dinaturality or problems we had with let-ccc's?}\smnt{Better now?
 % Otherwise, can you suggest a formulation?}

\subsection{Definition and Examples of Guarded Fixpoint Operators}

\begin{assumption} \label{mainassumption}
  We assume throughout the rest of the paper that $(\catC,\ibox)$ is a pair consisting of a category $\catC$ with finite
  products (also known as a \emph{cartesian category}) and a pointed endofunctor $\ibox: \catC \to
  \catC$, i.e.\ we have a natural transformation $p: \Id \to
  \ibox$. The endofunctor $\ibox$ is called \emph{delay} or the \emph{``later'' modality}.% endofunctor} of $\catC$.
\end{assumption}

\begin{remark}
In references like Birkedal et al. \cite{BirkedalMSS12:lmcs,BirkedalM13:lics}, much more is assumed about both the underlying category and the delay endofunctor. Modelling simply-typed lambda calculus requires cartesian closure. Dependent types require additional conditions like being a \emph{type-theoretic fibration category} (see, e.g., \cite[Definition IV.1]{BirkedalM13:lics}).   In such a case, one also wants to impose some limit-preservation or at least finite-limit-preservation condition on the delay endofunctor \cite[Definition 6.1]{BirkedalMSS12:lmcs}---e.g., to ensure the transfer of the guarded fixpoint operator to slices. We do not impose any of those restrictions %, not because they do not hold in our models, but simply
 because we do not need them in our derivations. For more on the connection with the setting of Birkedal et al. \cite{BirkedalMSS12:lmcs}, see Proposition \ref{prop:mgrt} below.
\end{remark}

%\tlnt{see the note above about preserving products/lax monoidality}
%\smnt{I'd say nothing because nothing more is assumed.}\tlnt{And yet, I decided to add something---do you think it can stay?}

\begin{definition}
  \label{def:dagger}
  A \emph{guarded fixpoint operator} on $(\catC,\ibox)$ is a family of
  operations
  \[
  \dagger_{X,Y} : \catC(\ibox X \times Y, X) \to \catC(Y,X)
  \]
  such that for every $f: \ibox X \times Y \to X$,
    \begin{equation}\label{eq:fixp}
      \vcenter{
        \xymatrix@C+1pc{
          Y 
          \ar[r]^-{\sol f}
          \ar[d]^{\qquad\quad \commu}_{\langle \sol f, Y\rangle }
          &
          X
          \\
          X \times Y
          \ar[r]_-{p_X \times Y}
          &
          \ibox X \times Y
          \ar[u]_{f}
        }
      }
    \end{equation}
    \iffull
    \[
    \inferrule{\Gamma, x: \ibox X \vdash F: X}{\Gamma \vdash \fix{x}{F} =  F [\delay{\fix{x}{F}}/x]}
    \]
    \fi
    where (as usual) we drop the subscripts and write $\sol f:Y \to X$ in lieu of
    $\dagger_{X,Y}(f)$. We call the triple  $(\catC,\ibox, \dagger)$ a \emph{guarded fixpoint category}.

    Moreover, $(\catC,\ibox,\dagger)$ is called a \emph{unique guarded fixpoint category} if
    for every $f: \ibox X \times Y \to X$, %there exists a unique
    $\sol f$ is the unique morphism satisfying~\refeq{eq:fixp}. In this
    case, we can just write a pair $(\catC,\ibox)$ instead of a triple $(\catC,\ibox,\dagger)$.
\end{definition}

%Usually, one either assumes that $\dagger$ satisfies further properties or even that $f^\dagger$ is the unique solution of~\refeq{eq:fixp}. 
If one does not require that $f^\dagger$ is the unique solution of~\refeq{eq:fixp}, one usually assumes that $\dagger$ satisfies further properties. For example, Simpson and Plotkin~\cite{sp00} require that a parametrized fixpoint operator $\dagger$ is natural in $Y$ in the base definition. 
%We shall see examples of unique
%$\dagger$ shortly, and 
We will come to the study of properties of $\dagger$ such as naturality in Section~\ref{sec:prop}. Let us begin with a list of examples. Note that in
most cases, we do not explicitly mention the action of $\ibox$ on morphisms whenever
it is canonical; for Example \refeq{ex:presheaves}, it is given in
Appendix A.

%We will now list several examples of categories with a guarded
%fixpoint operator. We will later see that they all satisfy the
%properties studied in the next section.

\begin{examples}\label{ex:cats}
  \begin{enumerate}[(1)]
    \item \label{ex:identity} Taking as $\ibox$ the identity functor on $\catC$ and $p_X$
      the identity on $X$ we arrive
      at the special case of categories with an ordinary fixpoint
      operator $\catC(X \times Y, X) \to \catC(Y,X)$ (see e.g.\
      Hasegawa~\cite{h97,h99} or Simpson and Plotkin~\cite{sp00}). Concrete
      examples are: the category $\cpob$ with its usual least fixpoint
      operator or (the dual of) any iteration theory of Bloom and
      \'Esik~\cite{be93}.  
   
    \item \label{ex:constant} Taking $\ibox$ to be the constant
      functor $\ibox X \deq 1$ %on the terminal
      %object $1$
      and $p_X \deq \mathord{!}: X \to 1$, %the unique morphism, 
      a trivial
      guarded fixpoint operator is given by the family of identity
      maps on the hom-sets $\catC(Y,X)$.  

%\tlnt{Shouldn't we mention in assumptions above rather than here that 1 denotes the
 % terminal object and ! the unique morphism? Also, in the next four examples  we don't say
 % anything about the action of the functor on morphism. Is it
 % fine?}\smnt{I agree re $1$ and $!$; can you please implement it. Re
 % action on morphisms, I think it is canonical in each case - so it is
 % fine with me; if you think this is too unclear for readers please
 % suggest formulation (note that for~(5) the def. on morph. is in
 % Appendix A).}
      
    \item \label{ex:cms} Take $\catC$ to be $\cms$, %of
      %complete 1-bounded metric spaces,
      $r \in (0,1)$, $\ibox_r: \cms \to \cms$ to be the endofunctor
      keeping the carrier of the space and multiplying all
      distances by $r$, and $p_X: X \to \ibox_r X$ to be the obvious
      ``contracted identity'' mapping. Note that a non-expansive
      mapping $f: \ibox_rX \to X$ is the same as an
      \emph{$r$-contractive} endomap, i.e.\ an endomap satisfying
      $d(fx,fy) \leq r \cdot d(x,y)$ for every $x, y \in X$. An
      application of Banach's unique fixpoint theorem yields a guarded fixpoint operator: for every $f: \ibox_r X \times Y \to X$ we consider the map 
\[
\Phi_f: \cms(Y,X) \to \cms(Y,X), \qquad \Phi_f(m) \deq f \cdot (p_X \times Y) \cdot \langle m, Y\rangle;
\]
notice that $\cms(Y,X)$ is a complete metric space with the $\sup$-metric 
\[
d_{Y,X}(m,n) \deq \sup_{y \in Y}\{d_X(my, ny)\}.
\]
Using that non-expansive maps are closed under composition, product and
pairing, it is easy to show that $\Phi_f$ is an $r$-contractive map, and so
its fixpoint is the unique non-expansive map $\sol f: Y \to X$
satisfying~\refeq{eq:fixp}.

    \item \label{ex:toptrees} Let $\cat A$ be a category with finite products and
       $\C$ be the presheaf category $\presh{\omega}{\A} \deq \A^{\omega^\op}$ of
      $\omega^{\mathsf{op}}$-chains in $\A$. The delay functor %$\ibox$ takes a
      %presheaf $X: \omega^\op \to \A$ to the presheaf $\ibox X$ with
      is given by
      $\ibox X(0) \deq 1$ and $\ibox X(n+1) \deq X(n)$ for $n \geq 0$,
      whereas $p_X$ is given by $(p_X)_0 \deq \mathord{!} : X(0) \to 1$ and $(p_X)_{n+1}
      \deq X (n+1 \geq n): X(n+1) \to X(n)$. For
      every $f: \ibox X \times Y \to X$ there is a unique $\sol f: Y
      \to X$ satisfying~\refeq{eq:fixp} given by %commutative it is defined
      %as follows: given $f: \ibox X \times Y \to X$ (i.e.\ $f_0: Y(0)
      %\to X(0)$ and $f_{n+1}: X(n) \times Y(n+1)\to X(n+1)$) one defines $\sol
      %f: Y \to X$ by 
      $\sol f_0 \deq f_0: Y(0) \to X(0)$ and 
      \[
      \sol f_{n+1}
      \deq
      (\xymatrix@1{
        Y(n+1) \ar[rrrr]^-{\langle \sol f_n \cdot Y(n+1 \geq n),
          Y(n+1)\rangle}
        &&&&
        X(n) \times Y(n+1)
        \ar[r]^-{f_{n+1}}
        &
        X(n+1)
      }).\]
     % It is not difficult to prove that $\sol f$ is the unique
      %morphism such that~\refeq{eq:fixp} commutes. 
      Notice that for $\A = \setc$, $\C$ is the ``topos of trees'' 
 of Birkedal et al. ~\cite{BirkedalMSS12:lmcs}.

      The next example generalizes this one. 

    \item
      \label{ex:presheaves} %Let $\catC$ be the topos of (contravariant) presheaves on a
      Assume $\WG \deq (W, \leq)$ is a well-founded poset, i.e,
      contains no infinite descending chains. As usual we write $x <
      y$ whenever $x \leq y$ and $x \neq y$. Furthermore, let $\catD$
      be a (small) complete category and $\catC \deq \presh{\WG}{\D}$,
      i.e., $\catC = \catD^{(W, >)}$. Define $\ibox X(w)$ to be the
      limit of the diagram whose nodes are $X(v)$ for $v < w$ and
      whose arrows are restriction morphisms: $\ibox X(w) \deq \lim_{v
        < w} X(v)$. As restriction mappings from $X(w)$ itself form a
      cone on that diagram, a natural $p_X: X \to \ibox X$ is given by
      the universal property of the limits. Note that for any minimal element $r \in W$, we have
      that $(\ibox X)(r)$ is the terminal object $1$ of $\catD$. The
      $\dagger$-operation on $f: \ibox X \times Y \to X$ is defined by
      induction on $(W, \leq)$: assuming that $\sol f_v$ is already defined for all $v < w$ let
      \[
      \sol f_w \deq (\xymatrix@1{
        Y(w) \ar[rr]^-{\langle k, Y(w)\rangle} && \ibox X(w) \times Y(w) \ar[r]^-{f_w} & X(w)
      }),
      \]
      where $k: Y(w) \to \ibox X(w)$ is the morphism uniquely induced by the cone 
      \[
      \xymatrix@1{
        Y(w) \ar[rr]^-{Y(w > v)} && Y(v) \ar[r]^-{\sol f_v} & X(v)
      }
      \qquad\text{for every $v < w$.}
      \]
      This includes the case of a minimal element $r$ for which the above definition yields $\sol f_r \deq
      f_r: Y(r) = 1\times Y(r) \to X(r)$.

      One can prove that $\sol f$ is a unique morphism of presheaves satisfying~\refeq{eq:fixp}; more details can be found in Appendix~\ref{sec:appA}. This result generalizes Birkedal et al. \cite [Theorem~2.4]{BirkedalMSS12:lmcs}.

       Regarding the examples given by Birkedal et al. \cite{BirkedalMSS12:lmcs}, see also Proposition \ref{prop:mgrt} below, which also establishes that a let-ccc with a fixpoint object with $\ibox$ and $\dagger$ just defined forms a guarded fixpoint category.

\item \label{ex:cpo} Let $\ibox$ be the lifting functor $(-)_\bot$ on
  $\cpo$, i.e.\ for any cpo $X$, $X_\bot$ is the cpo with a newly
  added least element and the natural transformation $p_X: X \to X_\bot$
  is the embedding of $X$ into $X_\bot$. Taking least fixpoints yields
  a guarded
  fixpoint operator. To see this
  notice that the hom-sets $\cpo(X,Y)$ are cpos with the pointwise
  order: $f \leq g$ iff $f(x) \leq g(x)$ for all $x \in X$. Now any
  continuous $f: X_\bot \times Y \to X$ gives rise to a continuous map
  $\Phi_f$ on $\cpo(Y, X_\bot)$:
  \[
  \Phi_f: \cpo(Y, X_\bot) \to \cpo(Y, X_\bot), \qquad
  \Phi_f(m) \deq p_X \cdot f \cdot\langle m, Y\rangle.
  \]
  Using the least fixpoint $s$ of $\Phi_f$  define
  $
  \sol f \deq (\xymatrix@1{
    Y 
    \ar[r]^-{\langle s, Y\rangle}
    &
    X_\bot \times Y
    \ar[r]^-f
    &
    X
  })$. 
  As $s = \Phi_f(s)$, one can easily show
  $\sol f$ satisfies~\refeq{eq:fixp}. %We will discuss the
  % details of this example below, in Section \ref{sec:fixmonad}.
  Just as Example (\ref{ex:presheaves}) was more general than~(\ref{ex:toptrees}),
  the present example is also known to be
  an instance of a more general situation:
  
\item \label{ex:letccc}  Crole and Pitts \cite{cp92,Crole:phd} define a
  \emph{let-ccc with a fixpoint object}%
  \takeout{
    \footnote{In fact, Crole and
      Pitts introduced \emph{let-ccc}'s and \cite{cp92} 
      postulates Definition 2.1 in a form which is equivalent to the
      one stated above only in a cartesian closed setting. However,
      the definition given herein is immediately suggested by 
      Lemma 2.2 \cite{cp92} and an explict claim theirein: \emph{if one simply has a
        category with finite products and a strong monad, the
        definition of fixpoint object should be strengthened to a
        parametrised form}. This is the route we follow
      here. Cf. Proposition \ref{prop:mgrt} below for an analogous
      discussion.} }
  as a tuple  
  \[
  (\catC, \mS,\eta,\mu, s, \mS\fxob \stackrel{\siin}{\longrightarrow} \fxob, 1 \stackrel{\om}{\longrightarrow} \fxob),
  \]
  where 
  \begin{itemize}
  \item $\C$ is a cartesian closed category,
  \item $(T,\eta, \mu, s)$ is a strong monad on $\C$, i.e.~a monad together with a \emph{strength} viz. a family of morphisms $s_{A,B}: A \times TB \to T(A \times B)$ natural in $A$ and $B$ and compatible with the monad structure in the obvious way,
  \item $\siin: T\fxob \to \fxob$ is an initial algebra for the functor $T$ and
  \item $\om$ is the equalizer of $\siin \cdot \eta_\fxob$ and
    identity on $\fxob$ (i.e., the unique fixpoint of $ \siin
    \cdot \eta_\fxob$) .
  \end{itemize}
  
  Recall from Crole and Pitts~\cite{Crole:phd,cp92} (and cf.~Manes~\cite{manes03}) that a strong monad can be specified by an object assignment $T$, a family of morphisms $\eta_A$ and an operation $\klei{(-)}: \C(A \times B, TC) \to \C(A \times TB, TC)$ satisfying the following laws:
  \begin{enumerate}[(a)]
  \item \label{kl:lif} for any $f: A \times B \to TC$, we have $\klei{f} \cdot (A \times \eta_B) = f$;
  \item \label{kl:pro} $\klei{(\eta_B \cdot \prr)} = \prr: A \times TB \to TB$;
  \item \label{kl:com} given $f: A \to A'$ and $g: A' \times B \to TC$ we have $\klei{(g \cdot (f \times B))} = \klei g \cdot (f \times TB)$;
  \item \label{kl:pai} given $f: A \times B \to TC$ and $g: A \times C \to TD$ we have
    $
    \klei{(\klei g \cdot \langle \prl, f\rangle)} = \klei g \cdot \langle \prl, \klei f\rangle.
    $
  \end{enumerate}
  
  The initiality of $\fxob$ and cartesian closedness of $\C$ yield that for any $f: C \times
  \mS A \to  A$ there exists a unique morphism $\iter{f}: C \times \fxob \to A$ such that 
  \begin{equation} \label{eq:defsiin}
    \vcenter{
      \xymatrix@C+3pc{
        C \times \fxob 
        \ar[r]^{\iter{f}}  
        \ar@{}[rd]|\commu
        & 
        A 
        \\
        C \times \mS \fxob 
        \ar[u]^{C \times \siin}
        \ar[r]_-{\langle \prl, \klei{(\eta \cdot \iter{f})}\rangle}
        & 
        C \times \mS A 
        \ar[u]_{f}
      }
    }
  \end{equation}
  
  Now set $\ibox \deq \mS$ and for any $f: Y \times TX \to X$
  define\footnote{Note that to be consistent with~\cite{cp92} we put
    $\ibox X$ in the right-hand product component of the domain of
    $f$.}  $\sol f \deq \iter{f} \cdot (Y \times \om)$. To see that this example covers the
  preceding one, follow \cite{cp92} and set $\mS \deq
  (-)_\bot$, $\eta_X \deq p_X$, $\klei{(-)}$ the strict extension (i.e.\ with $\klei f (a, \bot) = \bot$) and 
  $
  \fxob \deq \{ 0 \sqsubset 1 \sqsubset \dots  \sqsubset \top\}.
  $
  We return to let-ccc's in Theorem \ref{th:proplet} below.
\end{enumerate}
\end{examples}

\begin{remark}
   If we replaced initiality of $\fxob$ in Example~\ref{ex:cats}.\refeq{ex:letccc}
   by
   the  (parametrized) universal property \refeq{eq:defsiin} of $\fxob$ we could drop cartesian closure of $\C$ and simply assume $\C$ to be a category with products; cf.~Crole and Pitts \cite[Lemma 2.2]{cp92} and an explict claim therein: \emph{if one simply has a category with finite products and a strong monad, the definition of fixpoint object should be strengthened to a parametrised form}. See also Proposition~\ref{prop:mgrt} below for an analogous discussion.   
    \end{remark}

%\begin{remark} 
Birkedal et al. \cite{BirkedalMSS12:lmcs} provide a general setting for topos-theoretic examples like~\ref{ex:cats}.\refeq{ex:toptrees} and~\ref{ex:cats}.\refeq{ex:presheaves} (the latter restricted to the case of $\setc$-presheaves) by defining a notion of \emph{a model of guarded recursive terms} and showing that \emph{sheaves over complete Heyting algebras with a well-founded basis} proposed by Di Gianantonio and Miculan \cite{DiGianantonioM04:fossacs} are instances of this notion. The difference between Definition 6.1 in \cite{BirkedalMSS12:lmcs} and our  Definition \ref{def:dagger} is that in the former a) the delay endofunctor $\ibox$ is also assumed to preserve finite limits, in particular finite products. On other hand b) our equality \refeq{eq:fixp} is only postulated in the case where $Y$ is the terminal object, i.e., only a non-parametrized fixpoint identity is assumed but c) the dagger in this less general version of \refeq{eq:fixp} is assumed to be unique.  Now, one can show that assumptions a) and c) imply our parametrized identity \refeq{eq:fixp}  \emph{whenever the underlying category is cartesian closed}, in particular whenever $\C$ is a topos.  Let us state both the definition and the result formally:

\begin{definition}%[\cite{BirkedalMSS12:lmcs}]
  \label{def:mgrt} 
  A \emph{weak model of guarded fixpoint terms} is a triple $(\catC, \ibox, \unpdg)$, where 
  \begin{itemize}
  \item  $(\catC, \ibox)$  satisfy our general Assumption \ref{mainassumption},% i.e., $\ibox: \C \to \C$ is a pointed endofunctor (with point $p: \Id \to \ibox$) and $\catC$ has finite limits
   \item $\ibox$ preserves finite products with  $\iprod{X}{Y} : \ibox
     X \times \ibox Y \to \ibox(X \times Y)$ as the witnessing isomorphism,
  \item  $\unpdg$ is a family of
  operations $\unpdg_{X} : \catC(\ibox X, X) \to \catC(1,X)$
  such that for every $f: \ibox X \to X$, $\ssol f$ is a unique morphism for which %making the following square
  %commute:
    \begin{equation}\label{eq:mgrt}
      \vcenter{
        \xymatrix@C+1pc{
          1 
          \ar[r]^-{\ssol f}
          \ar[d]^{\qquad \commu}_{\ssol f }
          &
          X
          \\
          X 
          \ar[r]_-{p_X}
          &
          \ibox X
          \ar[u]_{f}
        }
      }
    \end{equation}
    %where (as usual) we drop the subscripts and write $\sol f:Y \to X$ in lieu of
    %$\dagger_{X,Y}(f)$.
\end{itemize}
\end{definition}

A \emph{model of guarded fixpoint terms} \cite{BirkedalMSS12:lmcs} is
a weak model in which in addition $\catC$  has finite \emph{limits}
(not just products) and $\ibox$ preserves them.

%The definition in \cite{BirkedalMSS12:lmcs} assumes furthermore that $\ssol f$ is the \emph{unique} morphism making diagram \ref{eq:mgrt} commute. However, even without this assumption we can show:

%We write  $\iprod{X}{Y} : \ibox X \times \ibox Y \to \ibox(X \times Y)$
%for the isomorphism provided by the assumption of limit preservation for the special case of product\footnote{One can note here that for the purpose of stating and proving Proposition \ref{prop:mgrt}, the assumption of finite limit preservation in Definition \ref{def:mgrt} can be weakened to finite product preservation. We only keep the stronger assumption for full consistency with \cite[Definition 6.1]{BirkedalMSS12:lmcs}.} of $X$ and $Y$. 
 
\begin{proposition} \label{prop:mgrt}
If $(\catC, \ibox, \unpdg)$ is a weak model of guarded fixpoint terms and $\catC$ is \emph{cartesian closed} with 
%\begin{tabular}{>{$}l<{$}@{\;$:$\;}>{$}l<{$}>{$}l<{$}}
%\curry{X}{Y}{Z}  & \catC(X \times Y, Z) & \to \catC(X, Z^Y), \\
%\uncurry{X}{Y}{Z}  & \catC(X, Z^Y) & \to \catC(X \times Y, Z), \\
%\eval{Y}{Z}   &  Z ^Y \times Y  & \to Z,
%\end{tabular}
\[
\curry{X}{Y}{Z}: \catC(X \times Y, Z) \to \catC(X, Z^Y),
\qquad\text{and}\qquad
\eval{Y}{Z}: Z^Y \times Y \to Z,
\]
then the operator $\dagger_{X,Y} : \catC(\ibox X \times Y, X) \to \catC(Y,X)$ defined as 
\[
%\uncurry{1}{Y}{X}
\sol f = \eval{Y}{X} \cdot \left(\left[\curry{\ibox(X^Y)}{Y}{X}\left(f \cdot \langle \ibox \eval{Y}{X} \cdot \iprod {X^Y}{Y} \cdot (\ibox(X^Y) \times p_Y), \prr \rangle\right)\right]^{\unpdg} \times Y\right)
%\xymatrix@1{Y \times \ibox(Y^X)}))%
\]
is a unique guarded fixpoint operator on $(\catC, \ibox)$. %(Note that
%we implicitly identified $Y$ and $1 \times Y$ above.)
%\tlnt{We do it so often elsewhere maybe it's worth including among our
 % conventions?}\smnt{I'd rather delete the ``(Note \dots)'' here. That
 % products are unique up to iso is standard category theory.}
\end{proposition}
\newcommand{\curr}[1]{\widehat{#1}}
\begin{proof}
Take $f: \ibox X \times Y \to X$. For notational simplicity, let 
\[
g \deq f \cdot \langle \ibox \eval{Y}{X} \cdot \iprod {X^Y}{Y} \cdot (\ibox(X^Y) \times p_Y), \prr \rangle
\qquad\text{and}\quad
%(\ibox \eval{Y}{X} \times Y) \cdot (\iprod {X^Y}{Y} \times Y) \cdot (\ibox(X^Y) \times \langle p_Y, \prr \rangle), 
\curr{g} \deq \curry{\ibox(X^Y)}{Y}{X}(g)
\]
\enlargethispage{10pt}
so that $\sol f = \eval{Y}{X} \cdot (\curr{g}^{\unpdg}  \times Y)$. We
need to show \refeq{eq:fixp} for $\sol f$ defined this way, i.e.,
commutativity of the outside of the diagram below: %(henceforth, we
%shall often drop subscripts of natural transformations whenever they
%are clear from the context): \tlnt{Again: a general
 % convention?}\smnt{I think it is fine to say this here where we first
 % use it.}
\[
\xymatrix@C+2.5pc{
  Y = 1 \times Y
  \ar@{}[rd]|(.4){\commu\ \text{by (\ref{eq:mgrt})}}
  \ar[r]^{ \curr{g}^{\unpdg} \times Y }
  \ar[d]_{\curr{g}^{\unpdg} \times Y }
  &  
  X^Y \times Y 
  \ar[rr]^{\ev} 
  \ar@{}[rd]|(.4){\commu\ \text{by ccc}}
  && 
  X 
  \\
  X^Y \times Y 
  \ar[r]^-{ p \times Y} 
  \ar[d]_{\langle \ev,\prr\rangle} 
  \ar[rrd]_(.4)*+{\labelstyle \langle p, \prr \rangle} 
  & 
  \ibox (X^Y)  \times Y 
  \ar[rru]^g
  \ar[u]^{ \curr{g} \times Y}  
  \ar[r]_-{\langle\ibox (X^Y) \times  p, \prr\rangle}  
  &  
  \ibox (X^Y) \times \ibox Y\times Y 
  \ar[d]^{ \can^{-1} \times Y} 
  & 
  \\
  X \times Y 
  \ar@{}[rr]_-{\commu\ \text{by pointedness}} 
  && 
  \ibox ( X^Y \times Y) \times Y 
  \ar[r]^-{ \ibox\ev \times Y} 
  \ar@{}[luu]|(.25){\langle(?), \commu\rangle}
  & 
  \ibox X \times Y \ar[uu]_f
  \ar@{<-} `d[l] `[lll]^{p \times Y} [lll]
  \ar@{}[uul]_{\commu\ \text{by def.~of $g$}}
} 
\]
Thus, we need to show the left-hand product component (?)~of the inner triangle commutes. We postcompose both sides with the isomorphism $\can:\ibox(X^Y \times Y) \to \ibox(X^Y)\times \ibox Y$ and obtain
\[
\can \cdot p_{X^Y \times Y} 
= 
\langle\ibox \prl, \ibox \prr\rangle \cdot p_{X^Y \times Y} 
= 
\langle p_{X^Y} \cdot \prl, p_Y \cdot \prr\rangle  = p_{X^Y} \times p_Y,
\]
where the middle equation holds by the naturality of $p$. 

\takeout{ % simplified above
And for the left-hand component we have
\[
\xymatrixcolsep{5pc}\xymatrix{
\ibox (X^Y) \times Y
 \ar[r]^{\ibox (X^Y) \times p_Y} 
 \ar@{}[rd]|(0.25){\commu} & \ibox (X^Y) \times \ibox Y
\ar@{=}[r] &  \ibox (X^Y) \times \ibox Y 
\ar[d]^{ \iprod{X^Y}{Y}}   \\
X^Y \times Y \ar[u]^{ p_{(X^Y)} \times Y}  
\ar[ru]|(0.4){p_{(X^Y)} \times p_Y}
\ar@{}[r]|(0.45){\commu\ \text{by pointedness}}
&   \ibox ( X^Y \times Y)  \ar[u]|{\ibox\prl \times
  \ibox\prr} 
\ar@{}[ru]|{\commu} 
\ar@{=}[r] &
 \ibox ( X^Y \times Y)  \\
X^Y \times Y 
\ar@{=}[u]^{\prl \times \prr}
\ar[ru]|{ p_{(X^Y \times Y)}} & &
}
\]
\[
\xymatrix{}
\]}

It remains to prove that $\sol f$ is the unique
solution of~\refeq{eq:fixp}. Suppose that $s: Y \to X$ satisfies $s = f \cdot (p_X \times Y) \cdot \langle s, Y\rangle$. Let $c \deq \curry{1}{Y}{X} (s): 1 \to X^Y$. We show that~\refeq{eq:mgrt} holds with $f^\ddagger$ replaced by $c$, i.e.~that $c = \curr g \cdot p_{X^Y} \cdot c$. Equivalently, we show
\[
\eval{Y}{X} \cdot (c \times Y) = \eval{Y}{X} \cdot (\curr g\times Y) \cdot (p_{X^Y} \times Y) \cdot (c \times Y). 
\]
In order to see this first modify the above diagram by replacing
$\curr g^\ddagger$ in the upper left-hand square by $c$. Notice that
our desired equation corresponds to the (modified) upper left-hand
square postcomposed with $\eval{Y}{X}$, i.e., the right-hand morphism of the
upper edge. Since $\eval{Y}{X} \cdot (c \times Y) = s$, the outside
of the modified diagram commutes by hypothesis. Thus, since all other
parts commute as indicated we obtain the desired equation. This
implies that $c = \curr g^\ddagger$, by the uniqueness of the
latter. Hence %, finally we conclude
\[
s = \eval{Y}{X} \cdot (c \times Y) =  \eval{Y}{X} \cdot (\curr g^\ddagger \times Y) = \sol f,
\]
which completes the proof.
\end{proof}

Proposition \ref{prop:mgrt} cannot be reversed: Example \ref{ex:cats}.(\ref{ex:cpo}) is a guarded fixpoint category, but $(-)_\bot$ clearly fails to preserve even finite products and hence it does not yield  a model of guarded recursive terms. 

\begin{example} \label{ex:bizjak}
A counterexample kindly provided by Ale\v{s} Bizjak shows that Proposition  \ref{prop:mgrt} does not hold without the assumption that $\catC$ is cartesian closed: there are examples of models of guarded recursive terms which are not guarded fixpoint categories. In fact, the adjective ``guarded'' can be removed altogether from  the previous sentence: $\catC$ in the Bizjak counterexample  is the category of groups and $\ibox = \Id_{\catC}$.  Initial and final objects coincide in $\catC$, which entails two things: 1) cartesian closure fails and 2) there is exactly one canonical choice for $\unpdg$ possible; $f^\unpdg:1 \to X$ is the unique morphism from the zero object for every group endomorphism $f: \ibox X = X \to X$. However, picking $X$ to be any non-trivial abelian group shows there is no way of defining $\dagger_{X,X}$. It is enough to consider $h: \ibox X \times X = X \times X \to X$ as being simply the group operation $+$: Equation \refeq{eq:fixp} then yields $\sol h(x) + x = \sol h(x)$ for every $x \in X$ and using inverses one has $x = 0$ contradicting nontriviality of $X$.
\end{example}

However, to apply Proposition  \ref{prop:mgrt}, it is enough that $(\catC, \ibox, \unpdg)$ is a \emph{full subcategory} of a cartesian closed model of guarded recursive terms such that, moreover, the inclusion functor preserves products and $\ibox$.

%\tlnt{Shall I turn this last sentence into a proper proposition?}
%\smnt{Yes!}

%\end{remark}

\begin{example} \label{ex:monads}
Monads provide perhaps the most natural and well-known examples of
pointed endofunctors.  Among delay endofunctors
in Example \ref{ex:cats},~\refeq{ex:identity}, \refeq{ex:constant}, \refeq{ex:cpo} and~\refeq{ex:letccc} happen to be monads. % In fact, while
%the first two ones are rather trivial monads, \ref{ex:cpo} is a
%paradigm example of a \emph{fixpoint monad} of Crole and Pitts
%\cite{cp92}.  
 In \refeq{ex:cms}, i.e.~the $\cms$ example, the type $\ibox\ibox A
 \to \ibox A$ is still inhabited (by any constant mapping), but one
 can easily show that monad laws cannot hold whatever candidate for
 monad multiplication is postulated. In the remaining (i.e.,
 topos-theoretic) examples, monad laws fail more dramatically:
 $\ibox\ibox A \to \ibox A$ is not even always inhabited. 
 \end{example}

 Section
 \ref{sec:cim} below discusses a class of examples of guarded fixpoint categories with unique dagger, where the delay endofunctor arises from (a module for) a monad. 
%a subclass of monads which happen to be
 %delay endofunctors with unique dagger. 
Later on, in Example \ref{ex:uniq} and Theorem \ref{th:proplet}, we will return to Examples~\ref{ex:cats}.\refeq{ex:cpo}--\refeq{ex:letccc}: we will see that while they do not
 possess uniqueness, they do enjoy other properties introduced in
 Section \ref{sec:prop} below.

\subsection{Completely Iterative Theories}
\label{sec:cim}

In this subsection we will explain how categories with guarded
fixpoint operator capture a classical setting in which guarded
recursive definitions are studied---Elgot's (completely) iterative
theories~\cite{elgot75,ebt78}. The connection to guarded fixpoint
operators is most easily seen if we consider monads in lieu of Lawvere
theories, and so we follow the presentation of \emph{completely
  iterative monads} by Milius~\cite{m05}. First we recall details of their motivating example: infinite trees on a signature $\Sigma$, i.e.\ a sequence $(\Sigma_n)_{n < \omega}$ of sets of operation symbols with
prescribed arity $n$. A $\Sigma$-tree $t$ on a set $X$ of generators is a
rooted and ordered (finite or infinite) tree
whose nodes with $n>0$ children are labelled by $n$-ary operation
symbols from $\Sigma$ and a leaf is labelled by a constant symbol
from $\Sigma_0$ or by a generator from $X$. One considers systems of mutually recursive
equations of the form
\[
x_i \approx t_i(\vec x, \vec y) \qquad i \in I,
\]
where $X = \{x_i \mid i\in I\}$ is a set of recursion variables and
each $t_i$ is a $\Sigma$-tree on $X+Y$ with $Y$ a set of parameters
(i.e.\ generators that do not occur on the left-hand side of a
recursive equation). A system of recursive equations is \emph{guarded}
if none of the trees $t_i$ is only a recursion variable $x \in
X$. Every guarded system has a unique \emph{solution}, which assigns
to every recursion variable $x_i \in X$ a $\Sigma$-tree $\sol t_i(\vec
y)$ on $Y$ such that $\sol t_i(\vec y) = t_i[\sigma]$, where $\sigma$ is the substitution replacing each $x_j$ by $\sol t_j(\vec y)$. For a concrete example, let $\Sigma$ consist of a binary
operation symbol $\ast$ and a constant symbol $c$, i.e.\ $\Sigma_0 =
\{c\}$, $\Sigma_2 = \{\ast\}$ and $\Sigma_n = \emptyset$ else. Then
the following system
\[
  x_1 \approx  x_2 \ast y_1 \qquad x_2 \approx (x_1 \ast y_2) \ast c,
\]
where $y_1$ and $y_2$ are parameters, has the following unique solution:
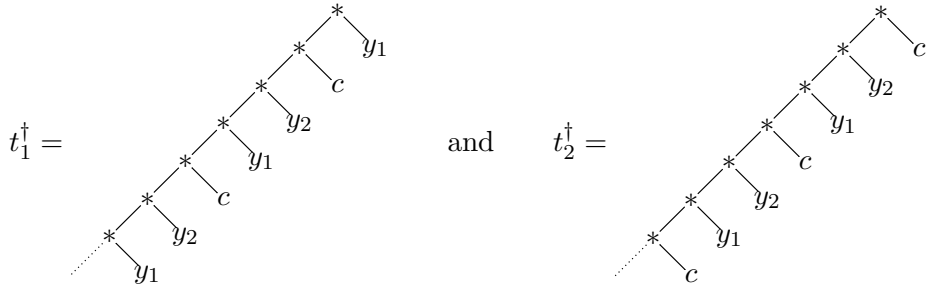
\[
\sol t_1 = 
\vcenter{
\xy
\POS (000,000) *{\ast} = "n1"
   , (-05,-05) *{\ast} = "n2"
   , (-10,-10) *{\ast} = "n3"
   , (-15,-15) *{\ast} = "n4"
   , (-20,-20) *{\ast} = "n5"
   , (-25,-25) *{\ast} = "n6"
   , (-30,-30) *{\ast} = "n7"
   , (-35,-35) = "n8"
%%%
\POS (005,-5) *{y_1}  = "m1"
   , (000,-10) *{c} = "m2"
   , (-05,-15) *{y_2} = "m3"
   , (-10,-20) *{y_1} = "m4"
   , (-15,-25) *{c} = "m5"
   , (-20,-30) *{y_2} = "m6"
   , (-25,-35) *{y_1} = "m7"
%%%
\POS "n1" \ar @{-} "n2"
\POS "n2" \ar @{-} "n3"
\POS "n3" \ar @{-} "n4"
\POS "n4" \ar @{-} "n5"
\POS "n5" \ar @{-} "n6"
\POS "n6" \ar @{-} "n7"
\POS "n7" \ar @{.} "n8"
%%%
\ar@{-} "n1";"m1"
\ar@{-} "n2";"m2"
\ar@{-} "n3";"m3"
\ar@{-} "n4";"m4"
\ar@{-} "n5";"m5"
\ar@{-} "n6";"m6"
\ar@{-} "n7";"m7"
\endxy
}
\qquad \textrm{and}\qquad 
\sol t_2 = 
\vcenter{
\xy
\POS (000,000) *{\ast} = "n1"
   , (-05,-05) *{\ast} = "n2"
   , (-10,-10) *{\ast} = "n3"
   , (-15,-15) *{\ast} = "n4"
   , (-20,-20) *{\ast} = "n5"
   , (-25,-25) *{\ast} = "n6"
   , (-30,-30) *{\ast} = "n7"
   , (-35,-35) = "n8"
%%%
\POS (005,-5) *{c}  = "m1"
   , (000,-10) *{y_2} = "m2"
   , (-05,-15) *{y_1} = "m3"
   , (-10,-20) *{c} = "m4"
   , (-15,-25) *{y_2} = "m5"
   , (-20,-30) *{y_1} = "m6"
   , (-25,-35) *{c} = "m7"
%%%
\POS "n1" \ar @{-} "n2"
\POS "n2" \ar @{-} "n3"
\POS "n3" \ar @{-} "n4"
\POS "n4" \ar @{-} "n5"
\POS "n5" \ar @{-} "n6"
\POS "n6" \ar @{-} "n7"
\POS "n7" \ar @{.} "n8"
%%%
\ar@{-} "n1";"m1"
\ar@{-} "n2";"m2"
\ar@{-} "n3";"m3"
\ar@{-} "n4";"m4"
\ar@{-} "n5";"m5"
\ar@{-} "n6";"m6"
\ar@{-} "n7";"m7"
\endxy}
\]

For any set $X$, let $T_\Sigma(X)$ be the set of $\Sigma$-trees on
$X$. It has been realized by Badouel~\cite{badouel89} that $T_\Sigma$
is the object part of a monad. A system of equations is then nothing
but a map
\[
f: X \to T_\Sigma(X+Y),
\]
and a solution is a map $\sol f: X \to T_\Sigma Y$ such that %the following square commutes:
\[
\xymatrix@C+2pc{
  X 
  \ar[r]^-{\sol f}
  \ar[d]^{\qquad\qquad \commu}_f
  &
  T_\Sigma Y
  \\
  T_\Sigma(X+Y) 
  \ar[r]_-{T_\Sigma[\sol f, \eta_Y]}
  &
  T_\Sigma T_\Sigma Y
  \ar[u]_{\mu_Y}
  }
\]
where $\eta$ and $\mu$ are the unit and multiplication of the monad
$T_\Sigma$, respectively. 

It is clear that the notion of equation and solution can be formulated
for every monad $S$. However, the notion of guardedness requires one
to speak about \emph{non-variables} in $S$. This is enabled by Elgot's
notion of \emph{ideal theory}~\cite{elgot75}, which for a finitary monad on $\setc$ is
equivalent to the notion recalled in the following definition. We
assume for the rest of this subsection that $\cat A$ is a category
with finite coproducts such that coproduct injections are
monomorphic. 

\begin{definition}[\cite{aamv03}]
  An {\em ideal monad} on $\A$ is a six-tuple
  $%\begin{displaymath}
    (S,\eta,\mu,S',\sigma,\mu')
  $ %\end{displaymath}
  consisting of a monad $(S,\eta,\mu)$ on $\A$, a subfunctor $\sigma:S'\subto S$
  and a natural transformation $\mu':S'S\to S'$ such that 
  \begin{enumerate}[(1)]
  \item $S=S'+\Id$ with coproduct injections $\sigma$ and $\eta$, and
  \item $\mu$ restricts to $\mu'$ along $\sigma$, i.e., 
    %the square below commutes:
    \begin{equation}
      \label{diag:muprime}
      \vcenter{
      \xymatrix{
        S'S
        \ar[0,2]^-{\mu'}
        \ar[1,0]^{\qquad\quad\commu}_{\sigma S}
        &
        &
        S'
        \ar[1,0]^{\sigma}
        \\
        SS
        \ar[0,2]_-{\mu}
        &
        &
        S
      }}
  \end{equation}
\end{enumerate}
\end{definition}
The subfunctor $S'$ of an ideal monad $S$ allows us to formulate the
notion of a guarded equation system abstractly; this leads to the
notion of completely iterative theory of Elgot et al.~\cite{ebt78} for
which we here present the formulation with monads from Milius~\cite{m05}:
\begin{definition}
  \label{def:eqsol}
  Let $(S,\eta,\mu,S',\sigma, \mu')$ be an ideal monad on $\A$.
  \begin{enumerate}
  \item An {\em equation morphism} is a morphism
    $ %\begin{displaymath}
      f:X\to S(X+Y)
    $ %\end{displaymath}
    in $\A$, where $X$ is an object
    (``of variables'') and $Y$ is an object (``of parameters'').
\item A {\em solution} of $f$ is a morphism $\sol{f}:X\to SY$ such that
  \begin{equation}\label{diag:solmon}
    \vcenter{
    \xymatrix@C+2pc{
      X
      \ar[r]^-{\sol{f}}
      \ar[d]^{\qquad\qquad\commu}_{f}
      &
      SY
      \\
      S(X+Y)
      \ar[r]_-{S[\sol{f},\eta_Y]}
      &
      SSY
      \ar[u]_{\mu_Y}
      }}
  \end{equation}
\item The equation morphism $f$ is called {\em guarded} if it factors
  through the summand $S'(X+Y) + Y$ of $S(X+Y) = S'(X+Y) + X + Y$:
  \begin{equation}\label{diag:guarded}
    \vcenter{
    \xymatrix@C+2pc{
      X
      \ar[r]^-{f}
      \ar @{.>} [rd]^{\qquad \commu}
      &
      S(X+Y)
      \\
      &
      S'(X+Y)+Y
      \ar[-1,0]_{[\sigma_{X+Y},\eta_{X+Y}\cdot\inr]}
      }}
    \end{equation}
  \item The given ideal monad is called \emph{completely iterative} if every guarded
  equation morphism has a unique solution. 
\end{enumerate}
\end{definition}

\begin{examples}
  We only briefly mention two examples of completely iterative
  monads. More can be found in literature~\cite{aamv03,m05,am06}. 
  \begin{enumerate}[(1)]
  \item
    The monad $T_\Sigma$ of $\Sigma$-trees is a completely iterative
    monad.  
  \item A more general example is given by parametrized final
    coalgebras. Let $H: \A \to \A$ be an endofunctor such that for
    every object $X$ of $\A$ a final coalgebra $TX$ for $H(-) + X$
    exists. 
    Then $T$ is the object assignment of a completely iterative
    monad; in fact, $T$ is the free completely iterative monad on
    $H$ (see~\cite{m05}). 
  \end{enumerate}
\end{examples}

We will now explain how a completely iterative monad $S$ yields a guarded fixpoint category.  Namely, let us show that the dual of
its Kleisli category $\C = (\A_S)^{op}$ is equipped with a guarded
fixpoint operator. First, since $\A_S$ has coproducts
given by the coproducts in $\A$, we see that $\C$ has products. Next we
need to obtain the endofunctor $\ibox$ on $\C$. This will be given as
the dual of an extension of the subfunctor $S': \A \to \A$ of $S$ to
the Kleisli category $\A_S$. Indeed, it is well-known that to have an
extension of $S'$ to $\A_S$ is equivalent to having a distributive law
of the functor $S'$ over the monad $S$ (see Mulry~\cite{mulry94}).

One easily verifies that the natural transformation
$\xymatrix@1{S'S \ar[r]^-{\mu'} & S' \ar[r]^-{\eta S'} & SS'}$ 
satisfies the two required laws and thus yields a distributive
law. The corresponding extension of $S'$ maps $X$ to $S'X$ and a morphism $X \to SY$ of $\A_S$ to 
\[
S'X \xrightarrow{S'f} S'SY \xrightarrow{\mu_Y'} S'Y \xrightarrow{\eta_{S'Y}} SS'Y.
\] 

Moreover, the endofunctor $\ibox^{op} = S'$ on $\A_S$ is
copointed, i.e.\ we have a natural transformation $p$ from $S'$ to $\Id: \A_S
\to \A_S$; indeed, its components at $X$ are given by the coproduct injections
$\sigma_X: S'X \to SX$, and it is not difficult to verify that this
is a natural transformation; thus, $\ibox$ is a pointed endofunctor on
$\C$. 

Now observe  that $\catC(\ibox X \times Y, X)$ is just $\A(X, S(S'X + Y))$. We are ready to describe the guarded fixpoint operator on $\C$. 

\begin{construction}
  \label{con:dagger}
  For any morphism $f: X \to S(S'X + Y)$ form the following morphism
  \[
  \ol f = (\xymatrix@1{
    X \ar[r]^-{f} & S(S'X + Y) \ar[rr]^-{S(\sigma_X + \eta_Y)} &&
    S(SX + SY) \ar[r]^-{S\can} & SS(X+Y) \ar[r]^-{\mu_{X+Y}} & S(X+Y)
  }),
  \]
  where $\can = [S\inl,S\inr]: SX + SY \to S(X+Y)$.  %\tlnt{Overloading
    %of the can notation?? apparently conscious}
  We shall verify that $\ol f$ is a guarded equation
  morphism for $S$ which allows defining $\sol f: X \to SY$ as the unique
  solution of $\ol f$. 
\end{construction}

\begin{proposition}
  For every $f$, $\sol f$ from Construction~\ref{con:dagger} is the
  unique morphism satisfying %$Y \to X$
         %                       in $\C$ such that
~\refeq{eq:fixp}.
  %commutes. 
\end{proposition}
\begin{proof} 
We first verify %\footnote{Henceforth we shall drop subscripts of
%               natural transformations whenever components are clear
 %              from the context.} 
that $\sol f$ is well-defined, i.e.~$\ol f$ is guarded with the following factor $\ol f_0$: 
  \[
  \xymatrix@C+.5pc{
    X 
    \ar[r]^-f 
    \ar[d]_{\ol f_0}
    & 
    S(S'X + Y) 
    \ar[r]^-{[\sigma, \eta]^{-1}} 
    & 
    S'(S'X + Y) + S'X + Y
    \ar[d]^{S'(\sigma + \eta) + S'X + Y}
    \\
    S'(X+Y) + Y
    &
    S'S(X+Y) + S'X + Y
    \ar[l]_(.65){\turnbox{30}{\ \ $\labelstyle [\mu', S'\inl] + Y$}}
    &
    S'(SX + SY) + S'X + Y
    \ar[l]_(.55){\turnbox{30}{\ $\labelstyle S'\can + S'X + Y$}}
  }
  \]
  Notice that both $\ol f$ and $\ol f_0$ start with $f$. Thus, in
  order to prove 
   $\ol f = [\sigma_{X+Y}, \eta_{X+Y}\cdot \inr] \cdot
    \ol f_0$  (see \refeq{diag:guarded}) we can remove $f$. Then it suffices to prove that the two remaining morphisms are equal when precomposed with the isomorphism $[\sigma, \eta]: S'(S'X + Y) + S'X + Y \to S(S'X + Y)$, i.e.~we verify
\[
\begin{array}{c}
[\sigma_{X+Y}, \eta_{X+Y} \cdot \inr] \cdot ([\mu_{X+Y}', S'\inl] + Y) \cdot (S'\can + S'X +Y) \cdot (S'(\sigma_X + \eta_Y) + S'X + Y) 
\\
= \mu_{X+Y} \cdot S\can \cdot S(\sigma_X + \eta_Y) \cdot [\sigma_{S'X + Y}, \eta_{S'X + Y}].
\end{array}
\]
We consider the two coproduct components separately and compute for the left-hand component
\[
\begin{array}{rcl}
  \sigma_{X+Y} \cdot \mu_{X+Y}' \cdot S'\can \cdot S'(\sigma_X + \eta_Y) 
  & \stackrel{\refeq{diag:muprime}}{=} & 
  \mu_{X+Y} \cdot \sigma_{S(X+Y)} \cdot S'\can \cdot S'(\sigma_X + \eta_Y) 
  \\
  & = &
  \mu_{X+Y} \cdot S\can \cdot S(\sigma_X + \eta_Y) \cdot \sigma_{S'X + Y},
\end{array}
\] 
where the second step uses naturality of $\sigma$ twice; for the right-hand component we obtain
\[
\begin{array}{rclp{4cm}}
  \mu_{X+Y} \cdot S\can \cdot S(\sigma_X + \eta_Y) \cdot \eta_{S'X + Y}
  & = &
  \mu_{X+Y} \cdot \eta_{S(X+Y)} \cdot \can \cdot (\sigma_X + \eta_Y) & (by nat.~of $\eta$) 
  \\
  & = & 
  \underbrace{[S\inl, S\inr]}_{= \can} \cdot (\sigma_X + \eta_Y) & (since $\mu \cdot \eta S = \id$)
  \\
  & = & 
  [S\inl \cdot \sigma_X, S\inr \cdot \eta_Y] 
  \\
  & = & 
  [\sigma_{X+Y} \cdot S'\inl, \eta_{X+Y} \cdot \inr]
  & (by nat.~of $\sigma$ and $\eta$)
  \\
  & = & 
  [\sigma_{X+Y}, \eta_{X+Y} \cdot \inr] \cdot (S'\inl + Y).
\end{array}
\]
This finishes the proof of guardedness of $\ol f$. 

Now observe that~\refeq{eq:fixp} can equivalently be expressed in $\A$ (using $\C =  (\A_S)^{op}$) as the outside square of the following diagram
\[
\xymatrix@C+1pc@R-.5pc{
  X
  \ar[rrr]^-{\sol f}
  \ar[dd]_f 
  \ar[rrrdd]_{\ol f}
  \ar@{}[rdd]_(.5)\commu
  &&&
  SY
  \ar@{}[ldd]_{\text{cf.~\refeq{diag:solmon}}}
  \\
  &&&
  SSY
  \ar[u]_{\mu}
  \\
  S(S'X + Y)
  \ar[r]_-{S(\sigma + \eta)}
  &
  S(SX + SY)
  \ar[r]_-{S\can}
  &
  SS(X+Y)
  \ar[r]_{\mu_{X+Y}}
  &
  S(X+Y)
  \ar[u]_{S[\sol f, \eta_Y]}
}
\]
This outside commutes iff the upper right-hand triangle commutes iff $\sol f$ is a solution of $\ol f$. Since the latter exists uniquely we see that $\sol f$ is the desired unique morphism satisfying~\refeq{eq:fixp}. 
\end{proof}
%
%In fact, to prove this proposition one shows that solutions of $\ol f: X \to
%S(X+Y)$ (i.e.\ morphisms $s: X \to SY$ such that~\refeq{diag:solmon}
%commutes) are in one-to-one correspondence with morphisms $Y\to X$ is
%$\C$ such that~\refeq{eq:fixp} commutes.
%
\begin{examples} \label{ex:uniq}
  Several items in Examples~\ref{ex:cats} are unique guarded fixpoint categories; this holds for Examples~\ref{ex:cats}.(\ref{ex:constant})--(\ref{ex:presheaves}), and also for the example of completely iterative monads in Section~\ref{sec:cim}. However, Example~\ref{ex:cats}.(\ref{ex:cpo}) is not a unique guarded fixpoint category: for let $X = \{0,1\}$ be the two-chain, $Y = 1$ the one element cpo and $f: X_\bot = X_\bot \times Y \to X$ be the map with $f(0) = f(\bot) = 0$ and $f(1) = 1$. Then both $0: 1 \to X$ and $1: 1 \to X$ make~\refeq{eq:fixp} commutative. %We discuss whether it meets other conditions now.
\end{examples}

\section{Properties of Guarded Fixpoint Operators}
\label{sec:prop}

In this section, we study properties of
guarded fixpoint operators. Except for uniformity, these properties are
purely equational. They are generalizing analogous 
properties of Bloom and \'Esik's iteration theories~\cite{be93}; more precisely, they would collapse to the original, unguarded counterparts when $\ibox$ is instantiated to the identity endofunctor as in Example~\ref{ex:cats}.\refeq{ex:identity}. Just like these original counterparts, they are all satisfied whenever the operator assigns a unique fixpoint (Theorem \ref{thm:unique}), but this  is not a necessary condition for them to hold, as witnessed by  Theorems \ref{th:proplet} and \ref{th:cpoconway} concerning  Examples~\ref{ex:cats}.\refeq{ex:cpo}--\refeq{ex:letccc}. However, the standard notion of dinaturality turns out to behave surprisingly enough in the guarded setting to merit a separate Subsection \ref{sec:dinat}. As a prerequisite, we also discuss the so-called Beki\v{c} identity in Section \ref{sec:bekic}.

\begin{definition}
  \label{def:prop}
   We
  define the following possible properties of  a guarded fixpoint
  category $(\catC,\ibox, \dagger)$:
%\tlnt{How about "possible properties"? I'd like to avoid conveying the impression that these properties are automatic.}\smnt{Fine.}

  \begin{enumerate}[{\bf (1)}]
  \item {\bf Fixpoint Identity ($\dagger$).} For every $f: \ibox X \times Y \to X$,~\refeq{eq:fixp} holds. This is built into the definition of guarded fixpoint categories and only mentioned here again for the sake of completeness.
  \item {\bf Parameter Identity (P).} For every $f: \ibox X \times Y \to X$ and
    every $h: Z \to Y$,
    \[
    \xymatrix@1{
      Z \ar[r]^-h & Y \ar[r]^{\sol f} & X
    }
    =
    (\xymatrix@1@C+1pc{
      \ibox X \times Z
      \ar[r]^-{\ibox X \times h}
      &
      \ibox X \times Y
      \ar[r]^-f
      &
      X
    })^\dagger.
    \]
      \item {\bf (Simplified) Composition Identity (C).} Given $f:\ibox X \times Y
    \to Z$ and $g: Z \to X$,
    \[
    (\xymatrix@1{
      \ibox X \times Y \ar[r]^-{f} & Z \ar[r]^-g & X
    })^\dagger
    =
    (\xymatrix@1{
      Y \ar[rr]^-{(f \cdot (\ibox g \times Y))^\dagger} && Z \ar[r]^-{g} & X
      }).
    \]
    
  \item {\bf Double Dagger Identity ($\dagger\dagger$).} For every $f: \ibox X \times \ibox X
    \times Y \to X$,
    \[
    (\xymatrix@1{
      Y \ar[r]^{f^{\dagger\dagger}} & X
    }) 
    =
    (\xymatrix@1{
      \ibox X \times Y \ar[r]^-{\Delta \times Y}
      &
      \ibox X \times \ibox X \times Y \ar[r]^-f & X
    })^\dagger.
    \]

  \item {\bf Uniformity (U).} Given $f: \ibox X \times Y \to X$, $g: \ibox X'
    \times Y \to X'$ and $h: X\to X'$,
    \[
    \vcenter{
      \xymatrix{
        \ibox X \times Y \ar[r]^-f \ar[d]^{\qquad\; \commu}_{\ibox h \times Y} & X \ar[d]^h \\
        \ibox X' \times Y \ar[r]_-g & X'
      }}
    \qquad
    \implies
    \qquad
    \vcenter{
      \xymatrix@R-1pc{
        & X \ar[dd]^h_{\commu \quad} \\
        Y
        \ar[ru]^-{\sol f}
        \ar[rd]_-{\sol g}
        \\
        & X'
      }}
    \]
  \end{enumerate}

  We call the first four properties (1)--(4) the \emph{Conway} axioms. 
\end{definition}

Notice that the Conway axioms are equational properties while (5) is
quasiequational, i.e.,\ an implication between equations.

Next we shall show that in the presence of certain of the above
properties the natural transformation $p: \Id \to \ibox$ is a derived
structure. Let $(\catC,\ibox)$ be equipped with an operator $\dagger$ %as in
%Definition~\ref{def:dagger} 
\emph{not necessarily satisfying~\refeq{eq:fixp}}. For every object $X$ of $\catC$ 
define $q_X: X \to \ibox X$ as follows: consider
\[
f_X = (\xymatrix@1@C+1.5pc{
  \ibox(\ibox X \times X) \times X \ar[r]^-{\ibox \prr \times X} &
  \ibox X \times X
})
\]
and form
\hfill
$
q_X = (\xymatrix@1{
  X \ar[r]^-{f_X^\dagger} & \ibox X \times X \ar[r]^-{\prl} & \ibox X
}).
$
\hspace*{\fill}

\begin{lemma} \label{lem:pointed}
  \ %Let $(\catC,\ibox)$ be equipped with the operator $\dagger$. Then:
  \begin{enumerate}
  \item If $\dagger$ satisfies the parameter identity and uniformity,
    then $q: \Id \to \ibox$ is a natural transformation.
  \item If $\dagger$ satisfies the fixpoint identity, then $q_X =
    p_X$ for all $X$. 
  \end{enumerate}
\end{lemma}
\begin{proof}
1.~For every morphism $h: X \to Y$ we have the following diagram:
\[
\xymatrix@C+2pc{
  X
  \ar[rr]^-{\sol f_X}
  \ar[d]_h
  \ar[rrd]|{\sol{(f_Y \cdot (\ibox (\ibox Y \times Y) \times h))}}
  &&
  \ibox X \times X
  \ar[r]^-{\pi_\ell}
  \ar[d]^{\ibox h \times h}
  \ar@{}[rd]|{\quad\commu}
  &
  \ibox X
  \ar[d]^{\ibox h}
  \ar@{<-} `u[l] `[lll]_-{q_X}^-\commu [lll]
  \\
  Y
  \ar[rr]_-{\sol f_Y}
  &&
  \ibox Y \times Y
  \ar[r]_-{\pi_\ell}
  &
  \ibox Y
  \ar@{<-} `d[l] `[lll]^-{q_Y}_-\commu [lll]
}
\] 
Of the left-hand square, the lower left-hand triangle commutes by the parameter identity and the upper right-hand triangle by uniformity since we have
\[
\xymatrix@C+2pc{
\ibox (\ibox X \times X) \times X
\ar[d]_{\ibox(\ibox h \times h) \times X}
\ar[rd]|{\ibox( \ibox h \times h) \times h}
\ar[rr]^-{f_X =}_-{\ibox \pi_r \times X}
&&
\ibox X \times X
\ar[d]^-{\ibox h \times h}
\\
\ibox (\ibox Y \times Y) \times X
\ar[r]_-{\ibox (\ibox Y \times Y) \times h}
\ar@{}[ru]^(.25)\commu
&
\ibox (\ibox Y \times Y) \times Y
\ar[r]_-{f_Y =}^-{\ibox \pi_r \times Y}
\ar@{}[u]|\commu
&
\ibox Y \times Y
}
\]

2.~Notice first that from the fixpoint identity for $\sol f_X$ we have: 
\[
\begin{array}{rcl}
  \pi_r \cdot \sol f_X & \stackrel{(\dagger)}{=} & \pi_r \cdot f_X \cdot (p_{\ibox X \times X} \times X) \cdot \langle \sol f_X, X\rangle \\
  & = & \pi_r \cdot (\ibox \pi_r \times X) \cdot (p_{\ibox X \times X} \times X) \cdot \langle \sol f_X, X\rangle \\
  & = & \id_X
\end{array}
\]
Then we consider the following diagram: 
\[
\xymatrix@-1pc{
  X
  \ar[rr]^-{\sol f_X}
  \ar[dd]_{\langle \sol f_X, X\rangle}
  \ar[rd]^{\Delta}
  &
  &
  \ibox X \times X
  \\
  &
  X \times X
  \ar[ru]_-{p \times X}
  \ar@{}[rd]^(.4)\commu_(.4){\text{(nat.~of $p$)}}
  \ar@{}[l]|(.65)\commu
  \\
  \ibox X \times X \times X
  \ar[ru]_(.6){\pi_r \times X}
  \ar[rr]_-{p \times X}
  &&
  \ibox (\ibox X \times X)\times X
  \ar[uu]_{\ibox \pi_r \times X}
}
\]
Since its outside commutes by the fixpoint identity so does the upper inner triangle. It follows that we have
$
q_X = \pi_\ell \cdot \sol f_X = \pi_\ell \cdot (p_X \times X) \cdot \Delta = p_X.
$
\end{proof}

\begin{definition}
   A guarded fixpoint category $(\catC,\ibox, \dagger)$  satisfying
  the Conway axioms (i.e.\ fixpoint, parameter, composition and double dagger identities) is called a \emph{guarded Conway category}.

  If in addition %the guarded fixpoint operator is uniform (i.e.\
  %satisfies uniformity) 
  uniformity is satisfied, we call $(\catC,\ibox,\dagger)$ a \emph{uniform guarded Conway
  category}. 
\end{definition}
%
%\begin{example}
  Note that an example of a guarded fixpoint category that is not a guarded Conway category and an example of a guarded Conway category that is not a uniform guarded Conway category exist already in the realm of iteration theories of Bloom and \'Esik (cf.~Examples~\ref{ex:cats}.\ref{ex:identity}). See \'Esik~\cite{e88}, Model~2 and Section~3, respectively. 
  %\end{enumerate}
%\end{example}

%In a unique guarded fixpoint category $(\catC,\ibox)$ the assignment $f \mapsto \sol
%f$ of the unique morphism such that~\refeq{eq:fixp} commutes obviously
%defines a guarded fixpoint operator $\dagger$ on $(\catC,\ibox)$ satisfying the fixpoint
%identities. 
%The next theorem states that such a unique $\dagger$ satisfies all
%the properties in Definition~\ref{def:prop}. 

\begin{theorem}
  \label{thm:unique}
  A unique guarded fixpoint category $(\catC,\ibox)$ is a
  uniform guarded Conway category. 
\end{theorem}
\begin{proof}%
  We shall prove that $\dagger$ satisfies the Conway axioms and uniformity.

  (1)~The fixpoint identity for $\dagger$ is satisfied by definition of a unique guarded fixpoint category.

  (2)~Parameter identity. Given $f: \ibox X \times Y \to X$ and $h: Z \to Y$ we have 
  \[
  \xymatrix@C+1pc{
    Z
    \ar[dd]_{\langle \sol f \cdot h, Z\rangle}
    \ar[r]^-h
    \ar@{}[rd]_(.6){\commu\ \ }
    &
    Y
    \ar[d]_{\langle \sol f, Y\rangle}
    \ar[r]^-{\sol f}
    &
    X
    \ar@{}[ld]|{\commu\ \text{by $(\dagger)$}}
    \\
    &
    X \times Y
    \ar[r]_-{p_X \times Y}
    \ar@{}[rd]_(.4){\commu}
    &
    \ibox X \times Y
    \ar[u]_-f 
    \\
    X \times Z
    \ar[ru]_{X \times h}
    \ar[rr]_-{p_X \times Z}
    &&
    \ibox X \times Z
    \ar[u]_{\ibox X \times h}
    }
  \]
  Thus, $\sol f \cdot h$ fits square~\refeq{eq:fixp} for $f \cdot (\ibox X \times h)$, and thus we have the desired equation by the uniqueness of $\sol{(f \cdot (\ibox X \times h))}$. 

(3)~Composition Identity. Let $f$ and $g$ be as in the definition of the identity. Then we have
\[
\xymatrix@C+1pc{
  Y
  \ar[r]^-{\sol{(f \cdot (\ibox g \times Y))}}
  \ar[dd]_{\langle \sol{(f \cdot (\ibox g \times Y))}, Y \rangle}
  &
  Z
  \ar[r]^-g
  &
  X
  \\
  &
  \ibox X \times Y
  \ar[u]^f
  \ar[r]_f
  \ar@{}[ru]|\commu
  &
  Z 
  \ar[u]_g
  \\
  Z \times Y
  \ar[r]_-{p_Z \times Y}
  \ar[d]_{g \times Y}
  \ar@{}[ruu]|{\commu\ \text{by ($\dagger$)}\quad}
  &
  \ibox Z \times Y
  \ar[u]^{\ibox g \times Y}
  \ar[rd]_-{\ibox g \times Y}
  \ar@{}[r]^(.6)\commu
  \ar@{}[d]_{\commu\ \text{by nat.~of $p$}}
  &
  \\
  X \times Y
  \ar[rr]_-{p_X \times Y}
  &&
  \ibox X \times Y
  \ar[uu]_f
}
\]
Since the outside commutes we obtain the desired equation by the unicity of $\sol{(g\cdot f)}$. 

(4)~Double Dagger Identity. Let $f: \ibox X \times \ibox X \times Y \to X$. Then we have
\[
\xymatrix@C+2pc{
  Y
  \ar[rr]^-{f^{\dagger\dagger}}
  \ar[ddd]_{\langle f^{\dagger\dagger}, Y\rangle}
  \ar@{}[rd]_{\commu\ \text{by ($\dagger$)}}
  &&
  X
  \\
  &
  \ibox X \times Y
  \ar[ru]^-{\sol f}
  \ar[d]^{\langle \sol f, \ibox X \times Y\rangle}
  \ar@{}[r]_(.6){\commu\ \text{by ($\dagger$)}}
  &
  \\
  &
  X \times \ibox X \times Y
  \ar[r]_-{p_X \times \ibox X \times Y}
  \ar@{}[rd]|\commu
  \ar@{}[u]^(.4){\langle (\ast), \commu, \commu \rangle}
  &
  \ibox X \times \ibox X \times Y
  \ar[uu]_f
  \\
  X \times Y
  \ar@/^1pc/[ruu]^-{p_X \times Y}
  \ar[ru]_(.6)*+{\labelstyle ((X \times p_X) \cdot \Delta) \times Y}
  \ar[rr]_-{p_X \times Y}
  &&
  \ibox X \times Y
  \ar[u]_{\Delta \times Y}
}
\]
We do not claim that part~($\ast$) commutes, but it commutes when precomposed with $\langle f^{\dagger\dagger} , Y\rangle$. This is because the lower passage yields simply $f^{\dagger\dagger}$ and the upper passage yields $\sol f \cdot (p_X \times Y) \cdot \langle f^{\dagger\dagger} , Y\rangle$, which is equal to $f^{\dagger\dagger}$ by the fixpoint identity. We conclude that the outside of the diagram commutes and so we obtain the desired equality by the unicity of $\sol{(f \cdot (\Delta \times Y))}$. 

(5)~Uniformity. Let $f$, $g$ and $h$ be as in the definition of uniformity. Then we have
\[
\xymatrix@C+1pc{
  Y
  \ar[rr]^-{\sol f}
  \ar[dd]_{\langle h \cdot \sol f, Y\rangle}
  \ar[rd]^{\langle\sol f, Y\rangle}
  &&
  X
  \ar[r]^-h
  \ar@{}[ld]_(.6){\commu\ \text{by ($\dagger$)}}
  &
  X'
  \\
  &
  X \times Y
  \ar[r]_-{p_X \times Y}
  \ar[ld]^{h \times Y}
  \ar@{}[rd]|{\commu\ \text{by nat.~of $p$}}
  \ar@{}[l]|(.65)\commu
  &
  \ibox X \times Y
  \ar[u]_f
  \ar[rd]^-{\ibox h \times Y}
  \ar@{}[ru]|\commu
  \\
  X' \times Y
  \ar[rrr]_-{p_{X'} \times Y}
  &&&
  \ibox X' \times Y
  \ar[uu]_g
}
\]
Since the outside commutes we obtain $h \cdot \sol f = \sol g$ by unicity of $\sol g$. 
\end{proof}
\begin{example}
  Coming back to Example~\ref{ex:uniq} we see that the unique $\dagger$ in 
  Examples~\ref{ex:cats}.(\ref{ex:constant})--(\ref{ex:presheaves} and the one for the example of completely iterative monads in Section~\ref{sec:cim} satisfy all the properties in Definition~\ref{def:prop}. But Example~\ref{ex:uniq} also
%that   $\catC = \cpo$ with the lifting functor $\ibox = (-)_\bot$  introduced as Example \ref{ex:cats}.\refeq{ex:cpo} does not satisfy uniqueness. Hence, a fortiori, 
entails that for let-ccc's with fixpoint objects of Example~\ref{ex:cats}.\refeq{ex:letccc} one cannot hope for uniqueness. So how about all the properties of Definition~\ref{def:prop}?
\end{example}
\begin{theorem}  \label{th:proplet}
  %\begin{enumerate}
  %\item 
    Let
    $
    (\catC, \mS,\eta,\mu, s, \mS\fxob \stackrel{\siin}{\longrightarrow} \fxob, 1 \stackrel{\om}{\longrightarrow} \fxob )
    $
    be a \emph{let-ccc with a fixpoint object}, let $\ibox \deq  \mS$ and for any $f: Y \times \mS X \to X$ define $\sol f \deq \iter{f} \cdot (Y \times \om)$, as in Example \ref{ex:cats}.\refeq{ex:letccc}. Then $\ibox$ satisfies all properties introduced in Definition~\ref{def:prop}, except, possibly, the double dagger identity ($\dagger\dagger$).
  %\item The special case of let-ccc with a fixpoint object  from Example \ref{ex:cats}.\refeq{ex:cpo} %, i.e., $(-)_\bot$ on $\cpo$ 
  %satisfies also ($\dagger\dagger$).
  %\end{enumerate}
\end{theorem}

We will see in Theorem \ref{th:cpoconway} below that for the special case where $\catC = \cpo$ and $T = (-)_\bot$ (see Example~\ref{ex:cats}.(\ref{ex:cpo})) more can be shown.
%
%Before we prove the theorem, we need a reminder on basic facts about
%and properties of let-cc's with a fixpoint object. First, the defining
%condition on $\iter{(-)}$ in Equation \refeq{eq:defsiin} above implies
%in particular that  \dots
%
\newcommand{\propfont}[1]{\textit{#1}}
\begin{proof}%[of Theorem \ref{th:proplet}] 
  Recall first the 4 axioms (a)--(d) of the operation $\klei{(-)}$ from Example~\ref{ex:cats}.\refeq{ex:letccc}. Further observe that the action of $T$ on a morphism $g: Z = 1 \times Z \to X$ is defined by
  \begin{equation}\label{eq:T}
    Tg = \klei{(\eta_X \cdot g)}: TZ = 1 \times TZ \to  TX.
  \end{equation}
  Notice that axioms~(a) and (d) imply for $f: B \to TC$ and $g: C \to
  TD$ (i.e.~the special case for $A = 1$) the usual extension laws  of Manes
  \cite[Definition~2.13]{manes03}:
  \begin{equation}\label{eq:kleisli}
    \klei f \cdot \eta_B = f
    \qquad\text{and}\qquad
    \klei{(\klei g \cdot f)} = \klei g \cdot \klei f.
  \end{equation}

  Finally, recall that for let-ccc's we write $\mS = \ibox$ on the right-hand side of product components. 
  
  %\medskip\noindent
  %1.
  ~\propfont{The fixpoint identity} ($\dagger$). Take $f: Y \times \mS X \to X$. Then:
  \[
  \vcenter{
    \xymatrix{
      Y = Y \times 1 
      \ar[rr]^-{Y \times \om}
      \ar[rd]_(.4){Y \times \om}
      \ar[dd]_{\langle \prl, \sol f \rangle}
      \ar@{}[rrd]^(0.6){\commu \text{ by def. of } \om}
      \ar@{}[ddr]_(0.4){\commu}
      & &
      Y \times \fxob 
      \ar[r]^{\iter{f}}
      &
      X       
      \ar@{<-} `u[l] `[lll]_{f^\dagger}^\commu [lll]
      \\
      &
      Y \times \fxob 
      \ar[r]_{Y \times \eta}
      \ar[ld]^{\langle \prl, \iter{f}\rangle}
      \ar@{}[rd]|{\commu \text{ by axiom (a)}}%\ref{kl:lif}}}
      &
      Y \times \mS\fxob 
      \ar[rd]|{\labelstyle\langle \prl,  \klei{(\eta \cdot \iter{f})}\rangle}
      \ar[u]_{Y \times \siin}
      \ar@{}[r]^(0.65){\commu \text{ by \refeq{eq:defsiin}}}
      &
      \\
      Y \times X \ar[rrr]_{Y \times \eta}
      &
      &
      &
      Y \times \mS X \ar[uu]_{f}
    }
  }
  \]
  \propfont{The parameter identity} (P). Take $f: Y \times \mS X \to X$ and $h: Z \to Y$, and define $g \deq f \cdot (h \times \mS X)$. Then
  \[
  \vcenter{
    \xymatrixcolsep{4pc}\xymatrix@C+1pc{
      Z = Z \times 1 
      \ar[rd]_{Z \times \om}
      % \ar[rdd]|{Z \times (\eta\cdot\om)}
      \ar[ddd]_{h = h \times 1}
      \ar`u[r]`[rrrrdd]^(.7){(f \cdot (h \times T X))^{\dagger}} [rrrrdd]
      \ar[rr]^{Z \times \om}
      &
      &
      Z \times \fxob
      \ar@/^2pc/[ddrr]^{\iter{g}}
      \ar@{}[rr]|(0.7){\commu}
      &
      &
      \\
      &
      Z \times \fxob
      \ar[r]^{Z \times \eta}
      %\ar@{=}[ru]
      \ar[rdd]_{h \times \fxob}
      \ar@{}[ddl]|(.4){\commu}
      \ar@{}[u]_{\commu \text{ by def. of } \om }
      &
      Z \times \mS \fxob
      \ar@{}[rd]|{\langle\commu, (\ast\ast)\rangle}
      \ar[d]|{h \times \mS\fxob}
      \ar[r]^{\langle \prl, \klei{(\eta_X \cdot \iter{g})}\rangle}
      \ar[u]^{Z \times \siin}_{\qquad \commu \text{ by \refeq{eq:defsiin} }}
      \ar@{}[ld]^(.4){\langle \commu, (\ast) \rangle }
      &
      Z \times \mS X
      \ar[d]|{h \times \mS X}
      \ar[dr]^{g}
      &
      \\
      &
      &
      Y \times \mS \fxob
      \ar[r]_{\langle \prl, \klei{(\eta_X \cdot \iter{f})}\rangle}
      \ar[d]|{Y \times \siin}^(.6){\qquad \commu \text{ by \refeq{eq:defsiin} }}
      &
      Y \times \mS X
      \ar[r]_{f}
      \ar@{}[ru]|(.35)\commu
      &
      X
      \\
      Y = Y \times 1
      \ar[rr]^{Y \times \om} 
      \ar`d[r]`[rrrru]_(.7){\sol f} [rrrru]
      &
      &
      Y \times \fxob
      % \ar`r[ru]+/u3pc/[rru]|{\iter{f}} 
      \ar@/_1pc/[rru]_{\iter{f}} 
      \ar@{}[rr]|(0.7){\commu}
      % \ar@{<-}`l[u]`[uu]+/u3pc/^{h \times \fxob}[uuu]
      &
      & 
    }
  }
  \]
  The part $(\ast)$ by itself does not need to commute, but it does when precomposed with $\om$. The task reduces then to showing $(\ast\ast)$, viz.~the equation
\[
 \klei{(\eta_X \cdot \iter{f})} \cdot (h \times \mS \fxob) =
 \klei{(\eta_X \cdot \iter{f \cdot (h \times \mS X)})}.
\]
By axiom~(c), this reduces to showing
\[
 \klei{(\eta_X \cdot \iter{f} \cdot (h \times \fxob))} =
 \klei{(\eta_X \cdot \iter{f \cdot (h \times \mS X)})}.
\]
To show this it is sufficient to prove
\[
\iter{f} \cdot (h \times \fxob) = \iter{f \cdot (h \times \mS X)}.
\]
A proof of this relies on $\iter{-}$ being the unique morphism
satisfying a suitable instance of \refeq{eq:defsiin}:
\[
\vcenter{
  \xymatrix@C+3pc{
    Z \times \fxob 
    \ar[r]^-{h \times \fxob}
    \ar@{}[rd]_\commu
    &
    Y \times \fxob
    \ar[r]^-{\iter f}
    \ar@{}[rd]|{\commu \text{ by \refeq{eq:defsiin}}}
    &
    X
    \\
    &
    Y \times \mS \fxob
    \ar[u]^{Y \times \sigma}
    \ar[r]_-{\langle \prl, \klei{(\eta \cdot \iter{f})} \rangle}
    &
    Y \times \mS X
    \ar[u]_f
    \\
    Z \times \mS \fxob
    \ar[uu]^{Z\times \siin}
    \ar[ru]^-{h \times \mS \fxob}
    \ar[rr]_-{\langle \prl, \klei{(\eta \cdot \iter{f} \cdot (h \times \fxob))}\rangle}
    &
    \ar@{}[u]_(.4){\langle \commu,(\star)\rangle}
    &
    Z \times \mS X
    \ar[u]_{h \times \mS X}    
    }
}
\]
and we get $(\star)$ by axiom~(c) again. %\ref{kl:com} again.

\propfont{Composition Identity} (C). Assume $f: Y \times \mS X 
    \to Z$ and $g: Z \to X$.  We want to show:
%    \[
%    (\xymatrix@1{
 %     \mS X \times Y \ar[r]^-{f} & Z \ar[r]^-g & X
%    })^\dagger
 %   =
  %  \xymatrix@1{
%      Y \ar[rr]^-{(f \cdot (\mS g \times Y))^\dagger} && Z \ar[r]^-{g} & X
 %     },
  %  \]
%which boils down to
\[
\iter{g \cdot f} \cdot (Y \times \om) 
= 
g \cdot \iter{f \cdot (Y \times \mS g)} \cdot (Y \times \om).
\]
For this, it is enough to show 
\[
\iter{g \cdot f}  
= 
g \cdot \iter{f \cdot (Y \times \mS g)}.
\]
We again use the fact that $\iter{-}$ is the unique morphism satisfying a suitable instance of \refeq{eq:defsiin}, which in this case is:
\[
\vcenter{
  \xymatrixcolsep{4pc}\xymatrix@C+1pc{
    Y \times \fxob 
    \ar[r]^-{\iter{f \cdot (Y \times \mS g)}}
    &
    Z 
    \ar[r]^-g
    \ar@{=}[rd]
    &
    X
    \\
    &
    Y \times \mS Z 
    \ar[u]^{f \cdot (Y \times \mS g )}
    \ar[rd]^{Y \times \mS g}
    \ar@{}[r]|\commu
    \ar@{}[l]|{\commu \text{ by \refeq{eq:defsiin}}}
    &
    Z \ar[u]_g
    \\
    Y \times \mS \fxob
    \ar[uu]^{Y \times \siin}
    \ar[ru]|{\langle \prl, \klei{(\eta \cdot \iter{f \cdot (Y \times \mS g)})}\rangle}
    \ar[rr]_{\langle \prl,  \klei{(\eta \cdot g \cdot \iter{f \cdot (Y \times \mS g)})}\rangle }
    &
    \ar@{}[u]|(0.25){\langle \commu, (\ast)\rangle}
    &
    Y \times \mS X 
    \ar[u]_f
  }
}
\]
For part~$(\ast)$ we compute
\[
\begin{array}{rcl@{\qquad}p{5cm}}
  \mS g \cdot \klei{\bigl(\eta_Z \cdot \iter{f \cdot (Y \times \mS g)}\bigr)} 
  & = &
  \klei{(\eta_X \cdot g)} \cdot \klei{\bigl(\eta_Z \cdot \iter{f \cdot (Y \times \mS g)}\bigr)} 
  & by~\refeq{eq:T}
  \\
  & = & 
  \klei{\bigl(\klei{(\eta_X \cdot g)}\cdot \eta_Z \cdot \iter{f \cdot (Y \times \mS g)}\bigr)}
  & by~\refeq{eq:kleisli} \\
  & = & 
  \klei{\bigl(\eta_X \cdot g \cdot \iter{f \cdot (Y \times \mS g)}\bigr)}
  & by~\refeq{eq:kleisli}.
\end{array}
\]
%which by naturality of $\eta$ can be replaced by
%\[
%\mS g \cdot \klei{(\eta_Z \cdot \iter{f \cdot (Y \times \mS g)})} 
%= 
%\klei{(\mS g \cdot \eta_Z \cdot \iter{f \cdot (Y \times \mS g)})}
%\]
%
%\tlnt{To conclude the proof, it seems I need exactly the same trick as the one you used to show uniformity, line \textit{(really?)} below. But why is this legal and valid?}

%\tlnt{Und was nun??? In my notes, I already wrote a tick next to it,
%  but now I cannot reproduce what I was thinking then. I cannot see any obvious argument that left and right should
%be equal!}

\propfont{Uniformity} (U).  Assume $f: Y \times \mS X \to X$, $g: Y \times \mS X' \to X'$ and $h: X \to X'$ are such that $h \cdot f = g \cdot (Y \times \mS h)$ holds. Our goal is to show $h \cdot \sol f = \sol g$, for which it is sufficient to show $\iter{g} = h \cdot \iter{f}$. Once again, we rely on  initiality property \refeq{eq:defsiin}, i.e., we need to show:
\[
\vcenter{
  \xymatrixcolsep{4pc}\xymatrix{
    Y \times \fxob
    \ar[r]^{\iter{f}}
    \ar@{}[rd]|{\commu \text{ by \refeq{eq:defsiin}}}
    &
    X 
    \ar[r]^h
    \ar@{}[rd]|{\commu \text{ by assumption}}
    &
    X'
    \\
    &
    Y \times \mS X
    \ar[u]_f
    \ar[dr]^{Y \times \mS h}
    &
    \\
    Y \times \mS \fxob
    \ar[ur]|{\langle \prl, \klei{(\eta \cdot \iter{f})}\rangle}
    \ar[uu]^{Y \times \siin}
    \ar[rr]_{\langle \prl,\klei{(\eta\cdot h \cdot \iter{f})}\rangle}
    &
    \ar@{}[u]|{\langle \commu, (\ast)\rangle}
    &
    Y \times \mS X'
    \ar[uu]_g
  }
}
\]
For $(\ast)$ we reason as follows:
\[
\begin{array}{rcl@{\qquad}p{5cm}}
  \mS h \cdot \klei{(\eta_X \cdot \iter{f})} 
  & = &  
  \klei{(\eta_{X'} \cdot h)} \cdot \klei{(\eta_X \cdot \iter{f})} 
  & by~\refeq{eq:T}
  \\
  & = &
  \klei{(\klei{(\eta_{X'} \cdot h)} \cdot \eta_X \cdot \iter{f})} 
  & by~\refeq{eq:kleisli}
  \\
  & = &
  \klei{(\eta_{X'} \cdot h \cdot \iter{f})} 
  & by~\refeq{eq:kleisli}.
\end{array}
\]
\end{proof}
%\tlnt{See above: it is not clear for me which of the axioms for  $\klei{(\cdot)}$ above makes \textit{really} legal}

%
\begin{theorem} \label{th:cpoconway}
  The category $\cpo$ with $\ibox = (-)_\bot$ and the dagger given by least fixpoint as in Example~\ref{ex:cats}.(\ref{ex:cpo}) satisfies all the properties of Definition~\ref{def:prop}.
\end{theorem}
\begin{proof}
 In the light of Theorem \ref{th:proplet}, we only need to show ($\dagger\dagger$). We use the notation of Example~\ref{ex:cats}.(\ref{ex:cpo}). In addition, for any $f: X_\bot \times Y \to X$, we define continuous functions $s_n: Y \to X_\bot$, $n \in \Nat$ as
\[
s_0 = \lambda y. \bot, 
\qquad  
s_{n+1} = \Phi_f(s_n)
\]
so that the least fixpoint of $\Phi_f$ is $s = \bigsqcup_{n \in \Nat} s_n$.
%it follows that 
%\begin{equation}\label{eq:lemfix}
%  p_X \cdot \sol f = p_X \cdot f \cdot \langle s, Y\rangle = s.
%\end{equation}

Now suppose we are given $f: X_\bot \times X_\bot \times Y \to X$. To prove ($\dagger\dagger$) we will first show that the least fixpoints $s$ of $\Phi_{\sol f}$ and $s'$ of $\Phi_{f \cdot(\Delta \times Y)}$ coincide, i.e.~we prove (a)~$s \sqsubseteq s'$ and (b)~$s' \sqsubseteq s$. 

For (a), it suffices to show that $s'$ is a prefixpoint of $\Phi_{\sol f}$, i.e.
\begin{equation}\label{eq:prefix}
  p_X \cdot \sol f \cdot \langle s', Y\rangle \sqsubseteq s'.
\end{equation}
To see this let $s''$ be the least fixpoint of $\Phi_f$. We will prove that 
\begin{equation}\label{eq:spp}
  s'' \cdot \langle s', Y\rangle = s'.
\end{equation}
by showing the two inequalities below by induction on $n$:
\begin{equation}\label{eq:twoineq}
  s_n'' \cdot \langle s', Y\rangle \sqsubseteq s' 
  \qquad\text{and}\qquad
  s_n' \sqsubseteq s'' \cdot \langle s', Y\rangle.
\end{equation}
Note that the left-hand inequalities above already imply~\refeq{eq:prefix} using that
\[
p_X \cdot \sol f = p_X \cdot f \cdot \langle s'',Y \rangle = \Phi_f(s'') = s''. 
\]
The right-hand inequalities in~\refeq{eq:twoineq} will be used at the end of our proof. 

For the induction proofs the base cases are obvious: $\bot \cdot \langle s', Y \rangle = \bot \sqsubseteq s'$ and $\bot \sqsubseteq s'' \cdot \langle s', Y\rangle$. For the induction steps we obtain
\[
\begin{array}{rcl@{\qquad}p{5cm}}
  s_{n+1}'' \cdot \langle s', Y\rangle 
  & = & p_X \cdot f \cdot \langle s_n'', X_\bot \times Y\rangle \cdot\langle s', Y\rangle 
  & since $s_{n+1}'' = \Phi_f(s_n'')$ \\
  & \sqsubseteq &  p_X \cdot f \cdot \langle s', X_\bot \times Y\rangle \cdot\langle s', Y\rangle 
  & by induction hypothesis \\
  & = & p_X \cdot f \cdot \langle s', s', Y\rangle \\
  & = & p_X \cdot f \cdot (\Delta \times Y) \cdot \langle s', Y\rangle \\
  & = & s' & since $s' = \Phi_{f \cdot (\Delta \times Y)}(s')$
\end{array}
\]
and
\[
\begin{array}{rcl@{\qquad}p{6cm}}
s_{n+1}' & = & p_X \cdot f \cdot (\Delta \times Y) \cdot \langle s_n', Y\rangle & since $s_{n+1}' = \Phi_{f\cdot (\Delta \times Y)}(s_n')$ \\
& = & p_X \cdot f \cdot \langle s_n', s_n', Y\rangle \\
& \sqsubseteq & p_X \cdot f \cdot \langle s''\cdot \langle s', Y\rangle, s', Y\rangle & by induction hypothesis and $s_n' \sqsubseteq s'$ \\
& = & p_X \cdot f \cdot \langle s'', X_\bot \times Y\rangle \cdot \langle s', Y\rangle \\
& = & s'' \cdot \langle s', Y\rangle & since $s'' = \Phi_f(s'')$.  
\end{array}
\]

For inequality (b) we prove by induction on $n$ that $s_n \sqsubseteq s'$ holds for all $n$. The base case is again trivial: $\bot \sqsubseteq s'$. For the induction step suppose that $s_n' \sqsubseteq s$. Then we consider the following diagram
\[
\xymatrix{
  Y
  \ar@/^7pt/[rr]^-{s_{n+1}'}
  \ar@/_7pt/[rr]_-{s}
  \ar@{}[rr]|{\begin{turn}{-90}$\labelstyle\sqsubseteq$\end{turn}}
  \ar@/_7pt/[d]_{\langle s_n', Y\rangle}
  \ar@/^7pt/[d]^{\langle s, Y\rangle}
  \ar@{}[d]|\sqsubseteq
  &
  &
  X_\bot
  \\
  X_\bot \times Y
  \ar[rr]^-{\sol f}
  \ar[rd]^(.6)*+{\labelstyle\langle \sol f, X_\bot \times Y\rangle}
  \ar[dd]_{\Delta \times Y}
  &
  \ar@{}[u]|(.45){\commu\ \text{since $s = \Phi_{\sol f}(s)$}}
  &
  X
  \ar[u]_{p_X}
  \\
  &
  X \times X_\bot \times Y
  \ar[ld]^(.4)*+{\labelstyle p_X \times X_\bot \times Y}
  \ar@{}[ru]_{\commu\ \text{by $(\dagger)$}}
  \ar@{}[l]|(.65){\langle (\ast), \commu, \commu\rangle}
  &
  \\
  X_\bot \times X_\bot \times Y
  \ar `r[rru]_(.6)f [rruu]
  &
  &
}
\]
Its outside commutes since $s_{n+1}' = \Phi_{f \cdot (\Delta \times Y)}(s_n')$ and~$(\ast)$ commutes when extended by $\langle s, Y\rangle$
%: we have
%\[
%s = \Phi_{\sol f}(s) = p_X \cdot \sol f \cdot \langle s, Y\rangle 
%\]
since $s$ is a fixpoint of $\Phi_{\sol f}$.  The  equalities in the diagram together with the inequality obtained from the induction hypothesis in the upper left-hand corner yield the desired inequality in the top row.

We are now ready to prove the desired equality $f^{\dagger\dagger} = (f \cdot (\Delta \times Y))^\dagger$:
\[
\begin{array}{rcl@{\qquad}p{5cm}}
f^{\dagger\dagger} & = & \sol f \cdot \langle s, Y\rangle & by def.~of $f^{\dagger\dagger}$ \\
& = & f \cdot \langle s'', Y\rangle \cdot \langle s, Y\rangle &  by def.~of $\sol f$ \\
& = & f \cdot \langle s'', Y\rangle \cdot \langle s', Y\rangle & since $s = s'$ \\
& = & f \cdot \langle s', s', Y\rangle & by~\refeq{eq:spp} \\
& = & f \cdot (\Delta \times Y) \cdot \langle s', Y\rangle \\
& = & \sol{(f \cdot (\Delta \times Y)} & by def.~of $\sol{(-)}$
\end{array} 
\]
This completes the proof.
\end{proof}
\begin{problem} \label{prob:letdd}
  Do let-ccc's with fixpoint objects with the dagger defined as in Example~\ref{ex:cats}.(\ref{ex:letccc}) satisfy the double dagger property $(\dagger\dagger)$?
\end{problem}
We do not see how an argument using two inequalities as in (a) and (b) as well as in~\refeq{eq:twoineq} generalizes to let-ccc's. However, as the flagship Example~\ref{ex:cats}.\refeq{ex:cpo} %standard example $\cpo$ with the lifting monad and dagger obtained by least fixpoint 
satisfies  $(\dagger\dagger)$, we believe that a counterexample might be intricate.

\subsection{The Beki\v{c} Identity}
\label{sec:bekic}

We generalize here the known fact that  the double dagger identity can be replaced by  the \emph{Beki\v{c} identity} (also known as the \emph{pairing identity}) among axioms of unguarded Conway theories  (see, e.g., \cite{Stefanescu87a},  \cite[Ch. 6.2, 6.8--6.9]{be93}, \cite{Esik:weighted}, \cite[Ch.~7.1]{h99}  and references therein). 
% Hasegawa~\cite{h97} used an equivalent version of the 
%\tlnt{Given the remark two pages below, the reader can ask if the idea/proof of replacing the ddi with Bekic should be ultimately credited to Bloom-Esik or to Hasegawa??}\smnt{I think Bloom and \'Esik, but Hasegawa was the one crediting Hyland; there exists no reference to Hyland -- I neutralized that remark.}
 We will make use of this  in our discussion of another property, dinaturality, in Section \ref{sec:dinat} and also in the discussion of trace operators in Section~\ref{sec:tr}.

\begin{definition}
    We introduce  the following possible property of a guarded fixpoint category $(\catC,\ibox, \dagger)$: 
  
 %\tlnt{Again, is "possible property" fine?}\smnt{Fine.}
  
  \medskip\noindent
  {\bf Beki\v{c} identity (B\v{c}).} For any $f: \ibox X \times \ibox Y \times A \to X$ and $g: \ibox X \times \ibox Y \times A \to Y$,
  \[
  \sol{\big(
      \xymatrix@1@C+1pc{
        \ibox (X \times Y) \times A \ar[r]^-{\can \times A} 
        &
        \ibox X \times \ibox Y \times A \ar[r]^-{\langle f, g\rangle}
        &
        X \times Y
      }
    \big)}
  =
  \langle \sol e_L, \sol e_R\rangle,
  \]
  where
  \begin{align*}
  e_R =&\ \big(\xymatrix@1@C+1pc{
      \ibox Y \times A \ar[rr]^-{\langle p_X \cdot \sol f, \ibox Y \times A\rangle}
      &&
      \ibox X \times \ibox Y \times A 
      \ar[r]^-g
      &
      Y
    }\big),
    \\
  e_L =&\ \big(\xymatrix@1@C+1pc{
      \ibox X \times A \ar[rr]^-{\ibox X \times \langle p_Y \cdot \sol e_R, A\rangle}
      &&
      \ibox X \times \ibox Y \times A
      \ar[r]^-f
      &
      X
    }\big).
  \end{align*}
\end{definition}

%The reader should bear in mind that  the unguarded  Beki\v{c} identity is ca
%
\begin{proposition}\label{prop:bekic}
   The Beki\v{c} identity holds in each  guarded Conway category $(\catC, \ibox, \dagger)$.
\end{proposition}
\begin{proof}
  First observe that $\can = (\ibox \prl \times \ibox \prr) \cdot \Delta$:
  \[
  \xymatrix@C+1pc{
    \ibox (X \times Y)
    \ar[r]^-{\ibox(\Delta)}
    \ar[rd]_-{\Delta}
    &
    \ibox ((X \times Y) \times (X \times Y))
    \ar[r]^-{\ibox (\prl \times \prr)}
    \ar[d]^{\can}
    \ar@{}[ld]^(.25)\commu
    \ar@{}[rd]^(.6)\commu
    &
    \ibox (X \times Y)
    \ar[d]^-{\can}
    \ar@{<-} `u[l] `[ll]_{\ibox (X \times Y)} [ll]
    \\
    &
    \ibox(X\times Y) \times \ibox(X\times Y)
    \ar[r]_-{\ibox \prl \times \ibox \prr}
    &
    \ibox X \times \ibox Y
    }
  \]
  Next we compute:
  \[
  \begin{array}{rcl}
    \sol{(\langle f, g\rangle \cdot (\can \times A))} 
    & = & \sol{(\langle f,g\rangle \cdot (\ibox \prl \times \ibox \prr \times A)\cdot (\Delta \times A))}
    \\
    & \stackrel{(\dagger\dagger)}{=} & 
    (\langle f,g \rangle \cdot (\ibox \prl \times \ibox \prr \times A))^{\dagger\dagger}
    \\
    & = &
    (\langle f,g \rangle \cdot (\ibox \prl \times \ibox Y \times A) \cdot (\ibox(X \times Y) \times \ibox \prr \times A))^{\dagger\dagger}
    \\
    & \stackrel{(P)}{=} &
    \sol{(\sol{(\langle f, g\rangle \cdot (\ibox \prl \times \ibox Y \times A))} \cdot (\ibox \prr \times A))}
  \end{array}
  \]
  Now let $h = \langle f, g \rangle \cdot (\ibox \prl \times \ibox Y \times A): \ibox(X \times Y) \times \ibox Y \times A \to X \times Y$. Then we have
  \[
  \begin{array}{rcl}
    \prl \cdot \sol h 
    & = & 
    \prl \cdot \sol{(\langle f, g \rangle \cdot (\ibox \prl \times \ibox Y \times A))}
    \\
    & \stackrel{(C)}{=} & 
    \sol{(\prl \cdot \langle f,g\rangle)}
    \\
    & = & \sol f. 
  \end{array}
  \]
  And we have
  \begin{equation}\label{diag:aux}
    \let\objectstyle=\labelstyle
    \vcenter{
      \xymatrix@C+2pc{
        \ibox Y \times A
        \ar[r]^-{\langle \sol h, \ibox Y \times A\rangle}
        \ar[rd]_-(.4)*+{\labelstyle\langle \sol f, \ibox Y \times A\rangle}
        &
        X \times Y \times (\ibox Y \times A)
        \ar[d]|{\prl \times \ibox Y \times A}
        \ar[r]_-{p \times (\ibox Y \times A)}
        \ar@{}[rd]|{\commu\ \text{by nat.~of $p$}}
        \ar@{}[ld]^(.25){\commu\qquad}
        &
        \ibox (X\times Y) \times \ibox Y \times A
        \ar[r]^-h
        \ar[d]|{\ibox \prl \times \ibox Y \times A}
        \ar@{}[rd]|\commu
        &
        X \times Y
        \ar[d]^{\prr}
        \ar@{<-} `u[l] `[lll]_{\sol h}^{\commu\ \text{by ($\dagger$)}} [lll]
        \\
        &
        X \times \ibox Y \times A
        \ar[r]_-{p \times \ibox Y \times A}
        &
        \ibox X \times \ibox Y \times A
        \ar[r]_-g
        &
        Y
      }
    }
  \end{equation}
  
 Plugging $h$ into our first computation above we obtain
  \[
  \begin{array}{rcl}
    \prr \cdot \sol{(\langle f, g\rangle \cdot (\can \times A))}
    & = &
    \prr \cdot \sol{(\sol h \cdot (\ibox \prr \times A))} 
    \\
    & \stackrel{(C)}{=} &
    \sol{(\prr \cdot \sol h)}
    \\
    & \stackrel{\refeq{diag:aux}}{=} & 
    \sol{(g \cdot (p_X \times \ibox Y \times A) \cdot \langle \sol f, \ibox Y \times A\rangle)}
    \\
    & = &
    \sol{(g \cdot \langle p_X \cdot \sol f, \ibox Y \times A\rangle)}
    \\
    & = &
    \sol e_R.
  \end{array}
  \]

  Let $z = \langle f,g \rangle \cdot (\ibox \prl \times \ibox \prr \times A)$. We saw previously that $\sol{(\langle f,g\rangle \cdot (\can \times A))} = z^{\dagger\dagger}$; thus we have
  \[
  \begin{array}{rcl}
    \prl \cdot \sol{(\langle f,g\rangle \cdot (\can \times A))} 
    & = &
    \prl \cdot z^{\dagger\dagger}
    \\
    & \stackrel{(\dagger)}{=} &
    \prl \cdot \sol z \cdot (p_{X \times Y} \times A) \cdot \langle z^{\dagger\dagger}, A\rangle
    \\
    & = &
    \prl \cdot \sol z \cdot \langle p_{X \times Y} \cdot z^{\dagger\dagger}, A \rangle
    \\
    & \stackrel{(P)}{=} & 
    \prl \cdot \sol{\big(z \cdot (\ibox(X \times Y) \times \langle p_{X \times Y} \cdot z^{\dagger\dagger}, A\rangle)\big)}
    \qquad (\ast)
  \end{array}
  \]
 Substitute the definition of $z$ and use that 
  \[
  \ibox \prr \cdot p_{X\times Y} \cdot z^{\dagger\dagger} 
  = 
  p_Y \cdot \prr \cdot z^{\dagger\dagger}
  =
  p_Y \cdot \prr \cdot \sol{(\langle f,g\rangle \cdot (\can \times A))}
  = 
  p_Y \cdot \sol e_R
  \]
  to obtain that $(\ast)$ is equal to
  \[
  \begin{array}{rcl}
    \multicolumn{3}{c}{
    \prl \cdot \big(\langle f, g \rangle \cdot (\ibox \pi_\ell \times \ibox \pi_r \times A) \cdot (\ibox(X \times Y) \times \langle p_{X \times Y} \cdot z^{\dagger\dagger}, A\rangle)\big)^\dagger
    }
    \\
    \qquad\qquad & = & \prl \cdot \sol{(\langle f, g \rangle \cdot (\ibox X \times \langle p_Y \cdot \sol e_R, A\rangle) \cdot (\ibox \prl \times A))}
    \\
    & \stackrel{(C)}{=} &
    \sol{(\underbrace{\prl \cdot \langle f, g \rangle}_{= f} \cdot (\ibox X \times \langle p_Y \cdot \sol e_R, A\rangle))}
    \\
    & = &
    \sol e_L.
  \end{array}
  \]
  This completes the proof.
\end{proof}

\begin{remark}
  Notice that the Beki\v{c} identity can also be formulated without
  mentioning $e_L$. In fact, by the parameter identity we have
  \begin{equation}\label{eq:eL}
    \sol e_L = (\xymatrix@1{
      A \ar[rr]^-{\langle p_Y \cdot \sol e_R, A\rangle}
      &&
      \ibox Y \times A 
      \ar[r]^-{\sol f}
      &
      X
    }).
  \end{equation}
\end{remark}

Proposition \ref{prop:bekic} states the Beki\v{c} identity can be derived from Conway laws.  But it can be also postulated directly as an axiom replacing  $(\dagger\dagger)$. This is a guarded counterpart of Proposition~5.3.15 in Bloom and  \'Esik~\cite{be93}:

%For ordinary fixpoint operators, it is known that  the Beki\v{c} identity can replace the double dagger law given other axioms of Conway theories~\cite[Proposition~5.3.15]{be93}. \tlnt{see comment above}\smnt{Ok now?} This holds also in the guarded setting:
  
%
\begin{proposition}
  Each guarded fixpoint category $(\C, \ibox, \dagger)$ satisfying the fixpoint, parameter, composition and Beki\v{c} identities is a guarded Conway category.
\end{proposition}
\begin{proof}
  We must derive $(\dagger\dagger)$ from the identities listed in the statement. Given $f: \ibox X \times \ibox X \times A \to X$ we apply the Beki\v{c} identity to obtain
  \[
  \sol{(\langle f, f\rangle \cdot (\can \times A))} = \langle \sol e_L, \sol e_R\rangle,
  \]
  where
  \hfill
  $
  e_R = f \cdot \langle p_X \cdot \sol f, \ibox X \times A\rangle 
  \qquad\text{and}\qquad
  e_L = f \cdot (\ibox X \times \langle p_X \cdot \sol e_R, A\rangle).
  $
  \hspace*{\fill}
  
  \noindent
  By the fixpoint identity we have $e_R = \sol f$. Thus, we obtain
  \[
  \begin{array}{rcl}
    f^{\dagger\dagger} & = & \sol e_R \\
    & \stackrel{\text{(B\v c)}}{=} & \prr \cdot \sol{(\langle f, f\rangle \cdot (\can \times A))} \\
    & = & \prr \cdot \sol{(\Delta_X \cdot f \cdot (\can \times A))} \\
    & \stackrel{\text{(C)}}{=} &
    \underbrace{\prr \cdot \Delta_X}_{= \id_X} \cdot\, \sol{(f \cdot (\can \times A) \cdot (\ibox (\Delta_X) \times A))}\\
    & = & \sol{(f \cdot (\Delta_{\ibox X} \times A))},
  \end{array}
  \]
  where the last equation follows from
  \[
  \can \cdot \ibox(\Delta_X) 
  = 
  \langle \ibox \prl, \ibox \prr\rangle \cdot \ibox(\Delta_X) 
  = 
  \langle \ibox(\id_X), \ibox (\id_X) \rangle 
  =
  \langle \id_{\ibox X}, \id_{\ibox X}\rangle
  = 
  \Delta_{\ibox X}.
  \]
  This completes the proof.
\end{proof}

\subsection{Dinaturality}
\label{sec:dinat}

Finally, we discuss a property that is essentially a parametrized version of the composition identity. In fact, Bloom and  \'Esik \cite{be93} use the very name \emph{composition identity} in this context, calling the unguarded counterpart of our earlier (C)  \emph{the simplified composition identity} instead. As it turns out, this property and its variants are not easy to understand in the guarded setting, leaving us with Open Problems \ref{prob:dinat} and \ref{prob:donetwo}.  But first, let us state basic notions and facts.

\begin{definition}
 We introduce  the following possible property of a guarded fixpoint category $(\catC,\ibox, \dagger)$:
  
  \medskip\noindent
  {\bf Dinaturality (D).} For every $f: \ibox X \times A \to Y$ and
  $g: \ibox Y \times A \to X$, 
  \[
  \qquad
  \sol{(\xymatrix@1{
    \ibox X \times A 
    \ar[rr]^-{\langle p_Y \cdot f, \prr\rangle}
    &&
    \ibox Y \times A
    \ar[r]^-g
    &
    X
  })}
  =
  \xymatrix@1{
    A 
    \ar[rr]^-{\langle p_Y \cdot \sol h, A\rangle}
    &&
    \ibox Y \times A
    \ar[r]^-g
    &
    X,
  }
  \]
  where
  $
  h = (\xymatrix@1{
    \ibox Y \times A \ar[rr]^-{\langle p_X\cdot g, \prr\rangle}
    &&
    \ibox X \times A \ar[r]^-{f}
    &
    Y
    }).
  $
\end{definition}

For unguarded fixpoint operators, it is well-known 
%Hyland \tlnt{ref??}\smnt{There is none - Hyland has this in unpublished notes.} has shown 
that the four Conway axioms are equivalent to dinaturality, the parameter and double dagger identities (D, P, $\dagger\dagger$), in other words, dinaturality can replace the fixpoint and composition identities (see, e.g., \cite{Stefanescu87a},  \cite[Ch. 6.2, 6.8--6.9]{be93}, \cite{Esik:weighted}, \cite[Ch.~7.1]{h99}  and references therein). %(see, e.g., Corollary 6.2.5 and Chapter 6.8 in  \cite{be93}, Chapter 7.1 in \cite{h97} ). 
Proposition~\ref{prop:dinat} below shows that we can derive dinaturality from the Conway axioms at the price of extra assumptions on $\ibox$. However, no extra assumptions are needed for:

\begin{proposition}
Dinaturality  holds in each  unique guarded fixpoint category $(\catC,\ibox)$.% then its $\dagger$ satisfies .
\end{proposition}
\begin{proof}
  Given $f$, $g$ and $h$ as in the definition of dinaturality, we only need to prove that $g \cdot \langle p_Y \cdot \sol h, A\rangle: A \to X$ satisfies the fixpoint identity \refeq{eq:fixp} w.r.t.~$g \cdot \langle p_Y \cdot f, \prr\rangle: \ibox X \times A \to X$. 
  Consider the diagram below:
  \[
  \xymatrix@C+1pc{
    A
    \ar[d]_{\langle p_Y \cdot \sol h, A, A\rangle}
    \ar[rr]^-{\langle p_Y \cdot \sol h, A\rangle}
    \ar[rrd]^{\langle \sol h, A\rangle}
    \ar[rd]|-{\langle p_Y \cdot \sol h, A\rangle}
    &&
    \ibox Y \times A
    \ar[r]^-{g}
    \ar@{}[ld]_\commu
    \ar@{}[rd]|\commu
    &
    X
    \\
    \ibox Y \times A \times A
    \ar[d]_{g \times A}
    \ar@{}[ru]|(.3)\commu
    &
    \ibox Y \times A
    \ar[l]^-{\ibox Y \times \Delta}
    \ar[rrd]_-{\langle p_X \cdot g, \prr\rangle}
    \ar@{}[r]|{\langle (\ast), \commu\rangle}
    \ar@{}[d]|\commu
    &
    Y \times A
    \ar[u]^{p_Y \times A}
    \ar[r]^{p_Y \times A}
    &
    \ibox Y \times A
    \ar[u]_g
    \ar@{}[ld]^(.2)\commu
    \\
    X \times A
    \ar[rrr]_-{p_X \times A}
    &&&
    \ibox X \times A
    \ar[u]_{\langle p_Y \cdot f, \prr\rangle}
    \ar[lu]_{\langle f, \prr\rangle}
    }
  \]
  For ($\ast$), recall $h = f \cdot \langle p_X \cdot g, \prr\rangle$ and then apply the fixpoint identity.
\end{proof}

\begin{proposition}\label{prop:dinat}
   Dinaturality holds  in each guarded Conway category $(\C, \ibox, \dagger)$  such that $\ibox$ preserves products and is well-pointed (i.e.~we have $\ibox p = p \ibox$).
\end{proposition}
\begin{proof}
  We prove this property from the fixpoint, composition and Beki\v{c} identities. Given $f: \ibox X \times A \to Y$ and $g: \ibox Y \times A \to X$, let $k = \langle p_Y \cdot f, \prr\rangle$. By~(C), we have
  \[
  \sol{(g \cdot k)} = (\xymatrix@1{
    A \ar[rr]^-{\sol{(k \cdot (\ibox g \times A))}} && \ibox Y \times A \ar[r]^-g & X
    }).
  \]
  Thus, in order to complete the proof it suffices to show that
  \[
  \langle p_Y \cdot \sol h, A\rangle 
  =
  \sol{(\xymatrix@1@C+1pc{
    \ibox(\ibox Y \times A)\times A
    \ar[r]^-{\ibox g \times A}
    &
    \ibox X \times A
    \ar[r]^-{k}
    &
    \ibox Y \times A
  })}.
  \]
  Since $\ibox$ preserves products, we have $\can^{-1}: \ibox\ibox Y \times \ibox A \to \ibox(\ibox Y \times A)$. Now let
  \begin{align*}
    m = &\ p_Y \cdot f \cdot (\ibox g \times A) \cdot (\can^{-1} \times A): \ibox \ibox Y \times \ibox A \times A \to \ibox Y, \\
    n = &\ \prr: \ibox \ibox Y \times \ibox A \times A \to A.
  \end{align*}
  Then we clearly have $\langle m,n\rangle \cdot (\can \times A) = k \cdot (\ibox g \times A)$. Now we apply the Beki\v{c} identity to obtain
  \[
  (\langle m,n\rangle \cdot (\can \times A))^\dagger = \langle \sol e_L, \sol e_R\rangle: A \to \ibox Y \times A,
  \]
  where
  \begin{align*}
    e_R \stackrel{\text{by def}}{=} &\ n \cdot \langle p_{\ibox Y} \cdot \sol m, \ibox A \times A\rangle = \prr: \ibox A \times A \to A \\
    e_L \stackrel{\text{by def}}{=} &\ m \cdot (\ibox\ibox Y \times \langle p_A \cdot \sol e_R, A\rangle): \ibox \ibox Y \times A \to \ibox Y.
  \end{align*}
 %\tlnt{Wait! You're using here def assign instead of ordinary equality symbol after all??}
 %\smnt{No, this is to indicate that the equation holds ``by definition'' of $e_R$ and $e_L$.}
  Using $e_R = \prr$ we see that 
  $
  \sol e_R \stackrel{(\dagger)}{=} e_R \cdot (p_A \times A) \cdot \langle \sol e_R, A\rangle = \id_A.
  $
  So we have 
  \[
  e_L = \underbrace{p_Y \cdot f \cdot (\ibox g \times A) \cdot (\can^{-1} \times A)}_m \cdot (\ibox \ibox Y \times \langle p_A, A\rangle).
  \]
  Thus we obtain
  $\
  \begin{array}[t]{rcl}
    \sol e_L 
    & \stackrel{(C)}{=} & p_Y \cdot \left(f \cdot (\ibox g \times A) \cdot (\can^{-1} \times A) \cdot (\ibox \ibox Y \times \langle p_A, A\rangle) \cdot (\ibox p_Y \times A)\right)^\dagger
    \\
    & = & 
    p_Y \cdot (f \cdot \langle p_X \cdot g, \prr\rangle)^\dagger,
  \end{array}
  $\\[5pt]
  where the second equation is derived as follows: it is sufficient to prove that the two morphisms inside $\dagger$ after removal of $f$ are equal, and for this one considers the product components of $\ibox X \times A$ (their codomain) separately. The right-hand component is obviously $\prr$ and the left-hand one follows from
  \[
  p_X \cdot g \stackrel{\text{(nat.~of $p$)}}{=} \ibox g \cdot p_{\ibox Y \times A} = \ibox g \cdot \can^{-1} \cdot (\ibox p_Y \times p_A),
  \]
  where the second equation is derived using well-pointedness of $\ibox$:
  $
  \can \cdot p_{\ibox Y \times A} = p_{\ibox Y} \times p_A = \ibox p_Y \times p_A.
  $
  %This completes the proof.
\end{proof}

%We leave as an open problem the question whether 

\begin{problem} \label{prob:dinat}
Do (D, P, $\dagger\dagger$) imply the fixpoint and simplified composition identities? 
\end{problem}

Further inspection reveals a curious asymmetry here. Under the assumption that $\ibox$ preserves products one can formulate two related versions of dinaturality where the given morphisms only contain one $\ibox$. For these properties we use for given $f: \ibox X \times A \to Y$ and $g: Y \times A \to X$ the morphism
\[
f \comp g = \big(\xymatrix@1@C+.2pc{
\ibox Y \times A 
\ar[rr]^-{\ibox Y \times \langle p_A, A\rangle} 
&&
\ibox Y \times \ibox A \times A
\ar[rr]^-{\can^{-1} \times A}
&&
\ibox(Y \times A) \times A
\ar[r]^-{\ibox g \times A}
&
\ibox X \times A
\ar[r]^-f
&
Y
}\big).
\]

\medskip\noindent
{\bf Property (D$_1$).} Given $f: \ibox X \times A \to Y$ and $g: Y \times A \to X$ we have
\[
\big(\xymatrix@1{
\ibox X \times A 
\ar[rr]^-{\langle f, \prr\rangle}
&&
Y \times A
\ar[r]^-g
&
X
}
\big)^\dagger
=
\big(
\xymatrix@1{
  A\ar[rr]^-{\langle \sol{(f \comp g)}, A\rangle}
  &&
  Y \times A
  \ar[r]^-g
  &
  X
}\big).
\]
{\bf Property (D$_2$).} Given $f: X \times A \to Y$ and $g: \ibox Y \times A \to X$ we have
\[
\sol{(g \comp f)} = g \cdot \langle p_Y \cdot \sol h, A\rangle,
\]
where $h = \big(\xymatrix@1{
\ibox Y \times A 
\ar[rr]^-{\langle g,\prr\rangle}
&&
X \times A
\ar[r]^-f
&
Y}\big).$

%Both versions are weaker than dinaturality above, i.e.,~
Whenever $\ibox$ is moreover well-pointed, each of (D$_1$) and (D$_2$) implies (D). One also readily proves, by adapting the proofs for unguarded operators, that (D$_1$) implies the simplified composition identity (C) and that (D$_2$) implies the fixpoint identity ($\dagger$). Conversely, the Conway axioms imply the first version of dinaturality (D$_1$). For the sake of brevity we leave the details to the reader. What defeats us at the moment is:

\begin{problem} \label{prob:donetwo} \ 
\begin{itemize}
\item Do the Conway axioms imply (D$_2$)? 
\item Does (D$_1$) imply (D$_2$)?
\end{itemize}
\end{problem}

%\subsection{Let-ccc's with a fixpoint object} \label{sec:fixmonad}

\ifbool{full}{
\subsection{Fixpoint expressions on types}

The starting point the fixpoint theorem in \cite{Sambin76:sl,Visser05:lncs} \dots For  importance of such results in contemporary type theory, see \cite{Nakano00:lics,BirkedalMSS12:lmcs,AtkeyMB13:icfp} \dots}{}

\takeout{ % we decided to take out the syntactic presentation
} % end takeout of syntactic presentation

%
% Guarded trace
%
\section{Guarded Trace Operators}
\label{sec:tr}

In the special case of Example~\ref{ex:cats}.(\ref{ex:identity}), i.e., $\ibox$ being the identity functor, it is well-known that a fixpoint operator 
satisfying the Conway axioms is equivalent to a trace operator w.r.t.\ the product on $\catC$~\cite{h97,h99}. In this section we present a similar
result for a %guarded fixpoint operators. To this end we begin by
generalized notion of a guarded trace operator on $(\catC, \ibox)$. 

\begin{remark} \label{rem:mono}
Recall that Joyal, Street and Verity~\cite{jsv96} introduced the unguarded notion of a trace operator for \emph{(symmetric) monoidal} categories. The applicability to non-cartesian tensor products is in fact one of main reasons of its popularity.  Our generalization can also be
formulated in the symmetric monoidal setting, see Remark \ref{rem:comono} below. %the remark preceding Construction \ref{con:inv} below. 
However, Theorems~\ref{thm:tr} and~\ref{thm:unif}, the main results in this section, do not make any use of this added generality. Hence, we keep the Assumption \ref{mainassumption} like in the remainder of the paper.
\end{remark}

\begin{definition}
  A (cartesian) \emph{guarded trace operator} on $(\catC, \ibox)$ is a natural
  family of operations
  \[
  \Tr_{A,B}^X: \catC(\ibox X \times A, X \times B) \to \catC(A,B)
  \]
  subject to the following three conditions:
  \begin{enumerate}
  \item {\bf Vanishing.} (V1) For every $f: \ibox 1 \times A \to B
    \cong 1 \times B$ we
    have
    \[
    \Tr_{A,B}^1(f) = (\xymatrix@1@C+1pc{
      A \cong 1 \times A
      \ar[r]^-{p_1 \times A}
      &
      \ibox 1 \times A
      \ar[r]^-f
      &
      B
      }).
    \]
    (V2) For every $f: \ibox X \times \ibox Y \times A \to X\times Y
    \times B$ we have
    \[
    \small
    \Tr_{A,B}^Y(\Tr_{\ibox Y \times A, Y \times B}^X(f)) =
    \Tr_{A,B}^{X\times Y} 
    (\xymatrix@1@C-.5pc{
      \ibox(X \times Y) \times A \ar[rr]^-{\can \times A}
      &&
      \ibox X \times \ibox Y \times A \ar[r]^-f
      &
      X \times Y \times B
      }).
    \]
  \item {\bf Superposing (S).} For every $f: \ibox X \times A \to X \times B$ we
    have
    \[
    \Tr_{A \times C, B \times C}^X (f \times C) = \Tr_{A,B}^X(f)
    \times C.
    \]
  \item {\bf Yanking (Y).} Consider the canonical isomorphism $c: \ibox X \times
    X \to X \times \ibox X$. Then we have
    \[
    \Tr_{X,\ibox X}^X(c) = (\xymatrix@1{X\ar[r]^{p_X} & \ibox X}).
    \]
  \end{enumerate}
  
  If $\Tr$ is a (cartesian) guarded trace operator on $(\catC, \ibox)$, $(\catC, \ibox, \Tr)$ is called
  a \emph{guarded traced (cartesian) category}.
\end{definition}

Of course, when $\ibox$ is taken to be the identity on $\C$ as in Example~\ref{ex:cats}.(\ref{ex:identity}), our notion of guarded trace specializes to the notion of an ordinary trace operator (w.r.t.~product) of Joyal, Street and Verity.  In addition, the  naturality of $\Tr$ can
equivalently be expressed by three more axioms, just like in the unguarded case:
\begin{enumerate}
  \setcounter{enumi}{3}
\item {\bf Left-tightening (Lt).} Given $f: \ibox X \times A \to X \times
  B$ and $g: A' \to A$ we have
  \[
  \Tr_{A',B}^X(\xymatrix@1{
    \ibox X \times A'
    \ar[rr]^-{\ibox X \times g}
    &&
    \ibox X \times A 
    \ar[r]^-f
    &
    X \times B
  })
  =
  (\xymatrix@1{
    A'\ar[r]^-g 
    & 
    A \ar[rr]^-{\Tr_{A,B}^X(f)}
    &&
    B
  }).
  \]
\item {\bf Right-tightening (Rt).} Given $f: \ibox X \times A \to X \times
  B$ and $g: B \to B'$ we have
  \[
  \Tr_{A,B'}^X(\xymatrix@1{
    \ibox X \times A \ar[r]^-f 
    &
    X \times B
    \ar[r]^-{X \times g}
    &
    X \times B'
  }) 
  = 
  (\xymatrix@1{
    A \ar[rr]^-{\Tr_{A,B}^X(f)} && B \ar[r]^-g & B'
  }).
  \]
\item {\bf Sliding (Sl).} Given $f: \ibox X \times A \to X' \times B$ and
  $g: X' \to X$ we have
  \[
  \Tr_{A,B}^X(\xymatrix@1{
    \ibox X \times A \ar[r]^-f 
    & 
    X' \times B
    \ar[r]^-{g \times B}
    &
    X \times B
  })
  =
  \Tr_{A,B}^{X'}(\xymatrix@1{
    \ibox X' \times A \ar[r]^{\ibox g \times A}
    &
    \ibox X \times A
    \ar[r]^-f
    &
    X ' \times B
  }).
  \]
\end{enumerate}

\begin{remark} \label{rem:comono}
  %As previously mentioned, the above definition of guarded trace can be
  %formulated 
  The generalization for a symmetric monoidal category $(\catC, \otimes, I,
  c)$ equipped with a pointed endofunctor $\ibox: \cat C \to \catC$
   requires the assumption that $\ibox$ is \emph{comonoidal}, i.e., equipped with a morphism
  $m_I: \ibox I \to  I$ and a natural transformation $m_{X,Y}: \ibox (X
  \otimes Y) \to \ibox X \otimes \ibox Y$ satisfying the usual coherence
  conditions. In fact, in the formulation of Vanishing (V2) we used
  that in every category the product $\times$ is comonoidal via $m_{X,Y} =
  \can$. 
\end{remark}

\iffull
\begin{example}
\tlnt{well \dots do we have any additional ones?} Suitably modified tracing on relations? Abramsky et al. on nuclear traces?
\end{example}
\fi

\begin{construction}\label{con:inv}
  \begin{enumerate}
  \item Let $(\catC, \ibox, \Tr)$ be a guarded traced category. Define a guarded
    fixpoint operator $\dace: \catC(\ibox X \times A, X) \to \catC(A,X)$ by
    \[
    f^{\dace} = \Tr_{A,X}^X(\xymatrix@1{
      \ibox X \times A \ar[r]^-{\langle f, f\rangle} & X \times X
    }): A \to X.
    \]
  \item Conversely, suppose $(\catC,\ibox, \dagger)$ is  a guarded
    fixpoint category. Define $\tragger_{A,B}^X:
    \catC(\ibox X \times A, X \times B) \to \catC(A,B)$ by setting for
    every $f: \ibox X \times A \to X \times B$
    \[
    \tragger_{A,B}^X (f) =
    (\xymatrix@1{
      A \ar[rr]^-{\langle \sol{(\prl \cdot f)}, A\rangle}
      &&
      X \times A
      \ar[rr]^-{p_X \times A}
      &&
      \ibox X \times A
      \ar[r]^-f
      &
      X \times B
      \ar[r]^-{\prr}
      &
      B 
    }). 
    \]
  \end{enumerate}
\end{construction}

%\enlargethispage{12pt}
%The main result in this section states that being
%guarded traced and guarded Conway are equivalent notions:
\begin{theorem}\label{thm:tr}
\begin{enumerate}
  \item Whenever $(\catC,\ibox, \Tr)$ is a guarded traced category,  $(\catC,\ibox,\dace)$ is a guarded Conway
  category. Furthermore, $\Tr_{\dace}$ is the original operator $\Tr$.
  \item Whenever $(\catC,\ibox, \dagger)$ is a guarded Conway category,  $(\catC,\ibox,\tragger)$ is guarded traced. Furthermore, $\dagger_{\tragger}$ is the original operator $\dagger$.
  \end{enumerate}
\end{theorem}

%More precisely, if $(\catC,\ibox)$ is guarded traced, then the guarded
%fixpoint operator defined in Construction~\ref{con:inv}.1 turns $(\catC,\ibox)$
%into a guarded Conway category. Conversely, if $(\catC,\ibox)$ is a guarded
%Conway category, then the operator $\Tr$ defined in
%Construction~\ref{con:inv}.2 turns $(\catC,\ibox)$ into a guarded traced
%category. Furthermore, the two constructions are mutually inverse. 

The proof details are similar to the unguarded case in Hasegawa~\cite{h99}. %Here one has to stick $\ibox$ in ``all the
%right places'' in all the necessary verifications of the axioms
%for trace and dagger, respectively.
As the derivation of the guarded version of the Beki\v{c} identity in  Proposition~\ref{prop:bekic}  has already shown,  it is  not a completely automatic adaptation. We give a complete proof in Appendix~\ref{sec:dagtr} below.

The process requires some creativity at times.

Hasegawa  related uniformity of trace to uniformity of dagger and we
can do the same in the guarded setup. Recall that in
iteration theories uniformity (called \emph{functorial dagger implication})
plays an important role. On the one hand, this quasiequation implies
the so-called \emph{commutative identities}, an infinite set of equational
axioms that are added to the Conway axioms in order to yield a complete
axiomatization of fixpoint operators in domains. On the other hand,
most examples of iteration theories actually satisfy uniformity, and
so uniformity gives a convenient sufficient condition to verify that a
given Conway theory is actually an iteration theory. 

%The following is a straightforward generalization of Hasegawa's notion
%of uniformity for a trace operator. 

\begin{definition}
  A guarded trace operator $\Tr$ is called \emph{uniform} if for every
 $f: \ibox X \times A \to X \times B$, $f': \ibox X' \times
  A \to X' \times B$ and $h: X \to X'$,
  \[
  \vcenter{
    \xymatrix{
      \ibox X \times A \ar[r]^-f
      \ar[d]^{\hspace{1cm}\commu}_{\ibox h \times A}
      &
      X \times B
      \ar[d]^{h \times B}
      \\
      \ibox X' \times A \ar[r]_-{f'}
      &
      X' \times B
    }}
  \qquad\implies\qquad
  \Tr_{A,B}^X(f) = \Tr_{A,B}^{X'}(f'): A \to B.
  \]
  A \emph{uniform guarded traced category} is a guarded traced category $(\catC, \ibox, \Tr)$ where $\Tr$ is uniform.
\end{definition}

%The two constructions in~\ref{con:inv} preserve uniformity:
\begin{theorem}\label{thm:unif}
\begin{enumerate}
  \item Whenever $(\catC,\ibox, \Tr)$ is a uniform guarded traced category, $(\catC, \ibox, \dace)$ is a uniform guarded Conway
  category. 
\item Whenever $(\catC,\ibox, \dagger)$ is a uniform guarded Conway category, $(\catC,\ibox, \tragger)$ is a uniform guarded traced category. 
  \end{enumerate}
\end{theorem}
The proof is in Appendix~\ref{app:unif}.

%Again, a more precise formulation would be: if $(\catC,\ibox)$ is a uniform guarded Conway category,
%then the guarded trace operator obtained in
%Construction~\ref{con:inv}.1 is uniform, too. And, conversely, if $(\catC,\ibox)$ is
%guarded traced with a uniform trace, then the $\dagger$ from
%Construction~\ref{con:inv} turns $(\catC,\ibox)$ into a uniform guarded Conway
%category. 

\begin{remark}%\smnote{Maybe leave this remark out?}
  Actually, Hasegawa proved a slightly stronger statement concerning
  uniformity than what we stated in Theorem~\ref{thm:unif}; he showed
  that a Conway operator is uniform w.r.t.~any fixed morphism $h: X
  \to X'$ (i.e.\ satisfies uniformity just for $h$) iff the
  corresponding trace operator is uniform w.r.t. this morphism $h$.  The proof is somewhat more complicated and in our guarded setting %of guarded trace and guarded dagger
   we leave this as an exercise to the
  reader. 
\end{remark}

Finally, let us note that the bijective correspondence between guarded
Conway operators and guarded trace operators established in
Theorem~\ref{thm:tr} yields an isomorphism of the \mbox{(2-)}categories of
(small) guarded Conway categories and guarded traced (cartesian)
categories. The corresponding notions of morphisms are, of course, as
expected:

\begin{definition}
  \begin{enumerate}
  \item $F: (\C,\ibox^\C,\dagger) \to (\D,\ibox^\D,\ddagger)$ is a morphism of guarded Conway categories whenever $F:
    \C \to \D$ is a finite-product-preserving functor satisfying
    \begin{equation}\label{eq:sat}
      \vcenter{
        \xymatrix{
          \C 
          \ar[r]^{\ibox^\C}
          \ar[d]^{\hspace{0.46cm}\commu}_F
          &
          \C
          \ar[d]^F
          \\
          \D \ar[r]_-{\ibox^\D}
          &
          \D
        }}
      \quad
      \text{and}
      \quad
      p^\D_{FX} = F(p_X^\C): FX \to \ibox^\D FX = F(\ibox^C X),
    \end{equation}
    and  preserving dagger, i.e., for every $f: \ibox X \times A \to X$ we
    have
    \[
    F(\sol f) = (\xymatrix@1{
      \ibox^\D FX \times FA \cong F(\ibox^\C X \times A) \ar[r]^-{Ff} & FX
    })^\ddagger.
    \]
  \item A morphism $F: (\C,\ibox^\C,\Tr_{\C}) \to (\D,\ibox^{\D},\Tr_{\D})$ of guarded traced categories is a finite-product-preserving $F: \C \to \D$ 
    satisfying~\refeq{eq:sat} above and preserving the trace
    operation: for every $f: \ibox^\C X \times A \to X \times B$ in
    $\C$ we have
    \[
    F(\Tr_{\C\,A,B}^{\;\;\; X}(f)) = \Tr_{\D\,FA,FB}^{\quad\! FX}(\xymatrix@1{
      \ibox^\D FX \times FA \cong F(\ibox^\C X \times A) \ar[r]^-{Ff}
      & F(X \times B) \cong FX \times FB
    }).
    \] 
  \end{enumerate}
\end{definition}

\begin{corollary}\label{cor:iso}
  The (2-)categories of guarded Conway categories and of guarded
  traced (cartesian) categories are isomorphic.
\end{corollary}
The proof is in Appendix~\ref{app:iso}.

\section{Conclusions and Future Work}
\label{sec:conc}

We have made the first steps in the study of equational
properties of guarded fixpoint operators popular in the recent literature~\cite{Nakano00:lics,Nakano01:tacs,AppelMRV07:popl,BentonT09:tldi,BirkedalMSS12:lmcs,KrishnaswamiB11:lics,KrishnaswamiB11:icfp,BirkedalMSS12:lmcs,AtkeyMB13:icfp,Litak14:trends}. We began with an extensive list of examples, including some not discussed so far as instances of delay endofunctors---e.g., Example \ref{ex:cats}.(\ref{ex:cpo}) or completely iterative theories in Section \ref{sec:cim}. Furthermore, we
formulated the four Conway axioms and the uniformity property in analogy to their unguarded counterparts and we showed their
soundness w.r.t.~the models discussed in Section \ref{sec:fix}. In
particular, Theorem \ref{thm:unique} proved that our axioms hold in all categories with
a unique guarded dagger. In Theorem \ref{thm:tr}, we have a generalization
of a result by Hasegawa for ordinary fixpoint operators: we proved
that to give a (uniform) guarded fixpoint operator satisfying the
Conway axioms is equivalent to giving a (uniform) guarded trace
operator on the same category.

Our paper can be considered as a work in progress report. %; with the
%results we presented here we merely probed that it makes sense to
%consider appropriately adjusted versions of the iteration theory axioms
%in the guarded setting. 
The long-term goal is to  arrive at completeness results similar to the ones for iteration theories. We do not claim that the axioms we presented are complete. In the unguarded setting, completeness is obtained by adding to the Conway axioms an infinite set of equational axioms called the \emph{commutative identities} ~\cite{be93,sp00}. We did not consider those here, but we considered the quasi-equational property of uniformity which implies the commutative identities and is satisfied in most models of interest. Only further research can show whether this property can ensure completeness in the guarded setup or one  needs to postulate stronger ones.

%A sequel paper will provide a syntactic type-theoretic presentation of the axioms we studied and a description of a classifying guarded Conway category. 

Let us recall  Open Problem \ref{prob:letdd} regarding soundness of ($\dagger\dagger$)  in the general setting of Crole and Pitts \cite{Crole:phd,cp92} and intriguingly complex status of guarded dinaturality leading to Open Problems \ref{prob:dinat} and \ref{prob:donetwo}.

%Concerning further models of guarded fixpoint operators, it would be worthwhile to consider fixpoint monads of Crole and Pitts~\cite{cp92} more closely. These generalize our example of the category $\cpo$ with the lifting monad. One can prove that any fixpoint monad induces a guarded fixpoint operator satisfying parameter and simplified composition identities as well as uniformity. However, proving the double dagger identity in the general case is an open problem. 
%Concerning let-ccc's of Crole and Pitts~\cite{cp92} we have shown that they yield guarded fixpoint categories satisfying the Conway axioms (except perhaps the double dagger identity) and uniformity. The standard example of cpo's also satisfies the double dagger identity, and so it would be worthwhile to resolve the issue of whether or not the double dagger identity holds for let-ccc's in general. 

It would also be interesting to study further examples of guarded traced monoidal categories which are not ordinary traced monoidal categories and which do not arise from guarded Conway categories. We have obtained some such examples but more work is needed to develop a full-blown theory. We postpone a detailed discussion to future work. %Traces w.r.t.~a trace ideal as considered by Abramsky, Blute and Panangaden~\cite{abp98} might be a good starting point.

\subparagraph*{Acknowledgements}

We would like to thank the anonymous referees whose comments have helped to improve the presentation of our paper. Thanks are also due to Ale\v{s} Bizjak for providing us with Example \ref{ex:bizjak}.
Besides, we would like to acknowledge an inspiring
discussion with Erwin R.~Catesbeiana  on \mbox{(un-)}productive \mbox{(non-)}termination. Finally, we have to credit in general William and Arthur for their very insistence on modal undertones in modern modelling of this major phenomenon (cf.~\cite{AppelMRV07:popl,wiki:pirates}).\smnote{I inserted the citation so that readers who are not Paul Andr\'e get the joke, too.} %and Tom and Nikolai for giving us the key to success in one word.
\tlnt{I like in general (in very model modern general) the idea of hidden jokes that only some initiated people would get. But if we're making this explicit, we should also quote the very original libretto, which I did now.}

%%
%% Bibliography
%%

%% Either use bibtex (recommended), but commented out in this sample

\bibliographystyle{fundam}
\bibliography{intmod,continuations,guarded-dagger}

%% .. or use bibitems explicitely

%\clearpage
\appendix
%\tlnt{appendix if needed \dots}

\section{Details for Example~\ref{ex:cats}.(\ref{ex:presheaves})}
\label{sec:appA}

First observe that $\ibox X$ is clearly a presheaf: for every $w' \leq w$ there exists a canonical morphism $\ibox X(w) = \lim_{v < w} X(v)  \to \lim_{v < w'} X(v) = \ibox X(w')$ induced by the universal property of the limit in the codomain; the functoriality easily follows from the uniqueness.

Next we define $\ibox$ on a morphism $f: X \to Y$ componentwise: for every $w$, $(\ibox f)_w$ is the unique morphism such that the following equations hold:
\[
\pi_v \cdot (\ibox f)_w = f_v \cdot \pi_v, \qquad (v < w), 
\]
where $\pi_v: X(w) = \lim_{v < w} X(v) \to X(v)$ denotes the limit projection. To see that $(\ibox f)_w$ is natural in $w$ it suffices to show that for any $w > w'$ the corresponding naturality square commutes when extended by the projection $\pi'_v: \ibox Y(w') \to Y(v)$ for every $v < w'$:
\[
\xymatrix@-.5pc{
  \ibox X(w)
  \ar[rrr]^-{(\ibox f)_w}
  \ar[dd]_{\ibox X(w > w')}
  \ar[rd]^{\pi_v}
  &&&
  \ibox Y (w)
  \ar[dd]^{\ibox Y(w >w')}
  \ar[ld]_-{\pi_v}
  \\
  &
  X(v)
  \ar[r]^-{f_v}
  \ar@{}[ru]|\commu
  \ar@{}[rd]|\commu
  \ar@{}[l]|(.6)\commu
  &
  Y(v)
  \ar@{}[r]|(.6)\commu
  &
  \\
  \ibox X(w')
  \ar[rrr]_-{(\ibox f)_{w'}}
  \ar[ru]^-{\pi_v'}
  &&&
  \ibox Y(w')
  \ar[lu]_-{\pi_v'}
}
%\begin{array}{rcl}
%  \pi'_v \cdot \ibox Y(w > w') \cdot (\ibox f)_w & = & \pi_v \cdot (\ibox f)_w %\\
%  & = & f_v \cdot \pi_v \\
%  & = & f_v \cdot \pi_v \cdot \ibox X(w > w') \\
%  & = & \pi'_v \cdot (\ibox f)_{w'} \cdot \ibox X(w > w').
%\end{array}
\]
A routine calculation then shows that $\ibox: \catC \to \catC$ is functorial.%\smnote{Insert details?}

The point $p: \Id \to \ibox$ is given componentwise as the unique morphism $(p_X)_w: X(w) \to \ibox X(w)$ such that $\pi_v \cdot (p_X)_w = X(w > v)$ for all $v < w$. Two easy routine calculations using the definitions of $\ibox X$ and $\ibox$ on morphisms, respectively, show that each component $p_X$ is natural in $w$ and that $p$ is natural in $X$. 

Let us now turn to the guarded fixpoint operator $\dagger$. We first prove simultaneously that each $\sol f_w$ is well-defined and that $\sol f$ is a morphism of $\catC$, i.e., $\sol f_w$ is natural in $w$. This is done by induction on $(W, \leq)$ (note that we do not have to distinguish the base case and induction step here). Fix any $w \in W$ and assume that $\sol f_v$ is well-defined for any $v < w$ and that the naturality condition $\sol f_{v'} \cdot Y(v > v') = X (v > v') \cdot \sol f_v$ holds for any $v' < v$ which are smaller than $w$. (Note that for a minimal $w \in W$ this holds trivially.) The latter naturality condition implies the cone property for $\sol f_v \cdot Y(w > v)$ inducing $k: Y(w) \to \ibox X(w)$ so that $\sol f_w$ is well-defined. We proceed to showing the naturality condition for any $w > w'$ using the following diagram (with $k'$ induced by the cone $\sol f_v \cdot Y(w' > v)$):
\[
\xymatrix@C+2pc{
  Y(w)
  \ar[d]_{Y(w > w')}
  \ar[r]^-{\langle k, Y(w)\rangle}
  \ar@{}[rd]_(.35){\langle (\ast), \commu\rangle}
  &
  \ibox X (w) \times Y(w) 
  \ar[r]^-{f_w}
  \ar[d]|{\ibox X(w > w') \times Y(w > w')}
  \ar@{}[rd]|{\commu}
  &
  X(w)
  \ar[d]^{X(w > w')}
  \ar@{<-} `u[l] `[ll]_{\sol f_w}^{\commu} [ll]
  \\
  Y(w')
  \ar[r]_-{\langle k', Y(w')\rangle}
  &
  \ibox X(w') \times Y(w')
  \ar[r]_-{f_{w'}}
  &
  X(w')
  \ar@{<-} `d[l] `[ll]^-{\sol f_{w'}}_-{\commu} [ll]
}
\] 
(Note that $\langle (\ast), \commu\rangle$ indicates that the right-hand product component of that part obviously commutes and the left-hand part, called $(\ast)$ is considered further.)
Part $(\ast)$ is seen commutative by extending it with the limit projection $\pi_v: Y(w') \to Y(v)$ for every $v < w'$ and performing a routine calculation. (Note again that this covers the cases where $w$ or $w'$ are minimal and consequently $k$ and $k'$, respectively, are the unique morphisms with codomain $1$.)

We are ready to verify the commutativity of~\refeq{eq:fixp}. This is done componentwise by induction on $(W, \leq)$. Assume that for a given $w$ all components of~\refeq{eq:fixp} at $v < w$ commute. Then we obtain for the $w$-component of $\sol f$ the following diagram (where $k$ is again induced by the cone $\sol f_v \cdot Y(w > v)$ and $h$ by the cone $f_v \cdot ((p_X)_v \times Y(v)) \cdot (\pi_v \times Y(w > v)): \ibox X(w) \times Y(w) \to X(v)$ for all $v < w$):
\[
\xymatrix{
  Y(w)
  \ar[r]^-{\langle k, Y(w)\rangle}
  \ar[d]_{\langle k, Y(w), Y(w)\rangle}
  &
  \ibox X(w) \times Y(w)
  \ar[r]^-{f_w}
  \ar@{=}[ddr]
  \ar@{}[dr]|(.7){\commu}
  \ar@{}[dd]|{\langle \mathrm{(ii)}, \commu\rangle}
  &
  X(w)
  \ar@{<-} `u[l] `[ll]_{\sol f_w}^{\commu} [ll]
  \\
  \ibox X(w) \times Y(w) \times Y(w)
  \ar[ru]_(.6)*+{\labelstyle h \times Y(w)}
  \ar[d]_{f_w \times Y(w)}
  \ar@{}[ru]^-(.35)*+{\labelstyle\langle \mathrm{(i)}, \commu\rangle\quad}
  &&
  \\
  X(w) \times Y(w)
  \ar[rr]_-{(p_X)_w \times Y(w)}
  &&
  \ibox X(w) \times Y(w)
  \ar[uu]_{f_w}
}
\]
Note that we are done if $w$ is minimal since $\ibox X(w) = 1$ is the terminal object. Otherwise for part~(i) we extend with the limit projection $\pi_v$ for every $v < w$ to obtain the following diagram (its outside commutes by the induction hypothesis, hence, so does part~(i) extended by $\pi_v$):
\[
\xymatrix@-.5pc@C-.5pc{
  Y(v)
  \ar[ddd]_{\langle \sol f_v, Y(v)\rangle}
  \ar[rrr]^-{\sol f_v}
  &
  \ar@{}[rd]|{\commu}
  &&
  X(v) 
  \\
  &
  Y(w)
  \ar[lu]^-*+{\labelstyle Y(w > v)}
  \ar[r]^-{k}
  \ar[d]_{\langle k, Y(w)\rangle}
  \ar@{}[rd]|(.3){\mathrm{(i)}}
  \ar@{}[ld]_(.6){\commu}
  &
  \ibox X(w) 
  \ar[ru]_-{\pi_v}
  \\
  &
  \ibox X(w) \times Y(w)
  \ar[ld]^(.4)*+{\labelstyle\pi_v \times Y(w > v)}
  \ar[ru]_h
  \ar@{}[rr]|{\commu}
  &&
  \\
  X(v)\times Y(v)
  \ar[rrr]_-{(p_X)_v \times Y(v)}
  &&&
  \ibox X(v) \times Y(v) 
  \ar[uuu]_{f_v}
}
\]
For part~(ii) observe first that $(p_X)_v \cdot \pi_v = \ibox X (w > v)$; indeed, this follows by routine calculation extending both sides by the limit projection $\pi_u: \ibox X(v) \to X(u)$ for every $u < v$. Now we obtain the commutativity of part~(ii) by extending it with every limit projection $\pi_v$:
\[
\xymatrix{
  \ibox X(w) \times Y(w) 
  \ar[rr]^-h
  \ar[ddd]_{f_w}
  \ar[rdd]^\commu_(.15){\raisebox{0pt}[5pt]{\turnbox{-40}{$\labelstyle\ibox X(w > v) \times Y(w > v)$}}}
  \ar[rd]^-(.6)*+{\labelstyle\pi_v \times Y(w > v)}
  &&
  \ibox X(w)
  \ar[dd]^{\pi_v}
  \\
  &
  X(v) \times Y(v)
  \ar[d]^{(p_X)_v \times Y(v)}
  \ar@{}[r]^(.6)\commu
  &
  \\
  \ar@{}[r]^(.4)\commu_(.4){\text{(nat. of $f$)}}
  &
  \ibox X(v) \times Y(v)
  \ar[r]^-{f_v}
  &
  X(v)
  \\
  X(w)
  \ar[rru]_-{X(w > v)}
  \ar[rr]_-{(p_X)_w}
  &
  &
  \ibox X(w)
  \ar[u]_{\pi_v}
  \ar@{}[lu]_(.4)\commu^(.4){\text{(def. of $p$)}}
}
\]

It remains to prove that $\sol f$ is unique such that~\refeq{eq:fixp} commutes. Suppose that $s: Y \to X$ is such that $s = f \cdot (p_X \times Y) \cdot \langle s, Y\rangle$. Then we prove that $\sol f = s$ componentwise by induction on $(W, \leq)$. Assume that $s_v = \sol f_v$ holds for all $v < w$. This implies that $k$ is the morphism induced by the cone $s_v \cdot Y(w > v) = X(w > v) \cdot s_w: Y(w) \to X(v)$. Hence, for all $v < w$ we have
\[
\pi_v \cdot k = X(w > v) \cdot s_w = \pi_v \cdot (p_X)_w \cdot s_w
\] 
from which we conclude that $k = (p_X)_w \cdot s_w$. (In the special case where $w$ is minimal this equation holds since it is an equation between morphisms with codomain $\ibox X(w) = 1$.) Thus, we obtain
\[
\begin{array}{rclp{5cm}}
  \sol f_w & = & f_w \cdot \langle k, Y(w)\rangle & (def.~of $\sol f_w$) \\
  & = & f_w \cdot \langle (p_X)_w \cdot s_w, Y(w)\rangle & (since $k = (p_X)_w \cdot s_w$) \\
  & = & s_w & (since $s = f \cdot (p_X \times Y) \cdot \langle s, Y\rangle$).
\end{array}
\]
This completes the proof.

\section{Proof of Theorem~\ref{thm:tr}}
\label{sec:dagtr}

The proof of Theorem~\ref{thm:tr} proceeds in three steps:
\begin{enumerate}[1.]
\item We show that $\dagger_\Tr$ defined in Construction~\ref{con:inv}.1 is a guarded trace operator. 
\item We show that $\Tr_\dagger$ defined in Construction~\ref{con:inv}.2 satisfies the Conway axioms.
\item We show that the two constructions are mutually inverse, i.e. $\dagger_{Tr_\dagger} = \dagger$ and $\Tr_{\dagger_\Tr} = \Tr$. 
\end{enumerate}

In the first two sections we shall drop the subscripts and only write $\Tr$ and $\dagger$ in lieu of $\Tr_\dagger$ and $\dagger_\Tr$, respectively. The proof is an adaptation of Hasegawa's proof for ordinary traces and fixpoint operators in~\cite{h99}. 

\subsection{From trace to dagger}
We will now prove that the $\dagger$-operation defined in Construction~\ref{con:inv}.1 satisfies the Conway axioms. But before we need an analogue of the fixpoint identity for traces:
\begin{lemma}\label{lem:fptr}
  For every $f: \ibox X \times A \to X \times B$ we form
  \[
  h = \Tr_{A,X}^X\big(\xymatrix@1{
    \ibox X \times A \ar[r]^-f & X \times B \ar[r]^-\prl & X \ar[r]^-\Delta & X \times X
  }\big).
  \]
  Then we have
  \[
  \Tr_{A,B}^X (f) = \big(\xymatrix@1@C+.5pc{
      A \ar[r]^-{\langle h, A\rangle}
      &
      X \times A
      \ar[r]^-{p_X \times A}
      &
      \ibox X \times A
      \ar[r]^-f
      &
      X \times B
      \ar[r]^-{\prr}
      &
      B
      }\big).
  \]
\end{lemma}
\begin{proof}
  Let $c: \ibox X \times X \to X \times \ibox X$ denote the canonical isomorphism swapping components. Observe that we have 
  \begin{equation}\label{eq:swap}
    c \cdot (p_X \times X) \cdot \Delta_X = c \cdot \langle p_X, X\rangle = \langle X, p_X\rangle = (X \times p_X) \cdot \Delta_X
  \end{equation}
  and
  \begin{equation}\label{eq:cc}
    (X \times c) \cdot (c \times X) \cdot (\ibox X \times \Delta_X) = (\Delta_X \times \ibox X) \cdot c.
  \end{equation}
  Now we compute
  \[
  \begin{array}{rcl}
    p_X \cdot h & = & p_X \cdot \Tr_{A,X}^X (\Delta_X \cdot \prl \cdot f) \qquad\text{by def.~of $h$} 
    \\
    & \stackrel{(Rt)}{=} &
    \Tr_{A,\ibox X}^X((X \times p_X) \cdot \Delta_X \cdot \prl \cdot f) 
    \\
    & \stackrel{\refeq{eq:swap}}{=} &
    \Tr_{A,\ibox X}^X(c \cdot (p_X \times X) \cdot \Delta_X \cdot \prl \cdot f)
    \\
    & \stackrel{(Y)}{=} &
    \Tr_{A,\ibox X}^X(c \cdot (\Tr_{X, \ibox X}^X(c) \times X) \cdot \Delta_X \cdot \prl \cdot f)
    \\
    & \stackrel{(S)}{=} & 
    \Tr_{A,\ibox X}^X(c \cdot \Tr_{X \times X, \ibox X \times X}^X(c \times X) \cdot \Delta_X \cdot \prl \cdot f)
    \\
    & \stackrel{(Lt)}{=} &
    \Tr_{A,\ibox X}^X(c \cdot \Tr_{\ibox X \times A, \ibox X \times X}^X((c \times X) \cdot (\ibox X \times (\Delta_X \cdot \prl \cdot f))))
    \\
    & \stackrel{(Rt)}{=} &
    \Tr_{A,\ibox X}^X(\Tr_{\ibox X \times A, X \times \ibox X}^X((X \times c) \cdot (c \times X) \cdot (\ibox X \times (\Delta_X \cdot \prl \cdot f))))
    \\
    & \stackrel{\refeq{eq:cc}}{=} &
    \Tr_{A,\ibox X}^X(\Tr_{\ibox X \times A, X \times \ibox X}^X((\Delta_X \times \ibox X) \cdot c \cdot (\ibox X \times (\prl \cdot f))))
    \\    
    & \stackrel{(V2)}{=} &
    \Tr_{A,\ibox X}^{X\times X}((\Delta_X \times \ibox X) \cdot c \cdot (\ibox X \times (\prl \cdot f)) \cdot (\can \times A))
    \\
    & \stackrel{(Sl)}{=} &
    \Tr_{A,\ibox X}^X(c \cdot (\ibox X \times (\prl \cdot f)) \cdot (\can \times A) \cdot (\ibox (\Delta_X) \times A))
    \\
    & = &
    \Tr_{A,\ibox X}^X(c \cdot (\ibox X \times (\prl \cdot f)) \cdot (\Delta_{\ibox X} \times A))
  \end{array}
  \]
  Using this we can finish the proof:
  \[
  \begin{array}{rcl}
    \prr \cdot f \cdot \langle p_X \cdot h, A\rangle 
    & = & 
    \prr \cdot f \cdot ((p_X \cdot h) \times A) \cdot \Delta_A
    \\
    & \stackrel{(S)}{=} &
    \prr \cdot f \cdot \Tr_{A\times A, \ibox X \times A}^X((c \times A) \cdot (\ibox X \times (\prl \cdot f) \times A) \cdot (\Delta_{\ibox X} \times A \times A)) \cdot \Delta_A
    \\
    & \stackrel{(Lt)}{=} &
    \prr \cdot f \cdot \Tr_{A,\ibox X \times A}^X((c \times A) \cdot (\ibox X \times (\prl \cdot f) \times A) \cdot ß\underbrace{(\Delta_{\ibox X} \times A \times A) \cdot (\ibox X \times \Delta_A)}_{= \Delta_{\ibox X} \times \Delta_A})
    \\
    & \stackrel{(Rt)}{=} &
    \Tr_{A,B}^X((X \times (\prr\cdot f)) \cdot (c \times A)\cdot (\ibox X \times (\prl \cdot f) \times A) \cdot (\Delta_{\ibox X} \times \Delta_A))
    \\
    & = &
    \Tr_{A,B}^X(\langle \prl \cdot f, \prr \cdot f\rangle)
    \\
    & = & 
    \Tr_{A,B}^X (f),
  \end{array}
  \]
  which completes the proof of the lemma.
\end{proof}

We now verify the Conway axioms for the $\dagger$-operation from Construction~\ref{con:inv}.1. 

(1)~Fixpoint identity. Given $f: \ibox X \times A \to X$ we apply Lemma~\ref{lem:fptr} to $\langle f, f\rangle$; then 
\[
h = \Tr_{A,X}^X(\Delta_X \cdot \prl \cdot \langle f,f\rangle) = \Tr_{A,X}^X (\langle f, f\rangle) = \sol f
\]
and therefore we have
\[
\begin{array}{rclp{5cm}}
  \sol f & = & \Tr_{A,X}^X\langle f, f\rangle & by def.~of $\dagger$ \\
  & = & 
  \prr \cdot \langle f,f\rangle \cdot (p_X \times A) \cdot \langle h, A\rangle
  & by Lemma~\ref{lem:fptr} 
  \\
  & = & 
  f \cdot (p_X \times A) \cdot \langle \sol f, A\rangle & since $h = \sol f$.
\end{array}
\]

(2)~Parameter identity. Let $f: \ibox X \times A \to X$ and $h: A' \to A$. Then we have
\[
\begin{array}{rclp{5cm}}
  \sol{(f \cdot (\ibox X \times h))}
  & = & \Tr_{A',X}^X(\langle f \cdot (\ibox X \times h), f \cdot (\ibox X \times h)\rangle)
  & by def.~of $\dagger$ \\
  & = & \Tr_{A',X}^X(\langle f, f \rangle \cdot (\ibox X \times h))
  \\
  & \stackrel{(Lt)}{=} & 
  \Tr_{A,X}^X (\langle f,f \rangle) \cdot h 
  \\
  & = & \sol f \cdot h & by def.~of $\dagger$.
\end{array}
\]

(3)~Composition identity. Given $f: \ibox X \times A \to Y$ and $g: Y\to X$ we compute
\[
\begin{array}{rclp{5cm}}
  \sol{(g \cdot f)} & = & \Tr_{A,X}^X(\langle g \cdot f, g \cdot f\rangle) & by def.~of $\dagger$ \\
  & = & \Tr_{A,X}^X((X \times g) \cdot \langle g\cdot f, f\rangle) \\
  & \stackrel{(Rt)}{=} & g \cdot \Tr_{A,Y}^X(\langle g\cdot f, f\rangle) \\
  & = & g \cdot \Tr_{A,Y}^X((g \times X) \cdot \langle f, f\rangle) \\
  & \stackrel{(Sl)}{=} & g \cdot \Tr_{A,Y}^Y(\langle f,f \rangle \cdot (\ibox g \times A)) \\
  & = & g \cdot \Tr_{A,Y}^Y(\langle f \cdot (\ibox g \times A), f \cdot (\ibox g \times A)\rangle) \\
  & = & g \cdot \sol{(f \cdot (\ibox g \times A))} & by def.~of $\dagger$.
\end{array}
\]

(4)~Double dagger identity. Given $f: \ibox X \times \ibox X \times A \to X$ we have
\[
\begin{array}{rclp{5cm}}
  f^{\dagger\dagger} & = & \Tr_{A, X}^X(\langle \sol f, \sol f\rangle) & by def.~of $\dagger$ \\
  & = & \Tr_{A, X}^X(\Delta_X \cdot \Tr_{\ibox X \times A,X}^X(\langle f, f\rangle)) & by def.~of $\dagger$ \\
  & \stackrel{(Rt)}{=} & 
  \Tr_{A, X}^X(\Tr_{\ibox X \times A, X \times X}^X((X \times \Delta_X) \cdot \langle f, f\rangle)) 
  \\
  & \stackrel{(V2)}{=} &
  \Tr_{A,X}^{X \times X}(\langle f, f, f \rangle \cdot (\can \times A))
  \\
  & = & 
  \Tr_{A,X}^{X \times X}((\Delta_X \times X) \cdot \langle f, f\rangle \cdot (\can \times A))
  \\
  & \stackrel{(Sl)}{=} &
  \Tr_{A,X}^X(\langle f, f\rangle \cdot \underbrace{(\can \times A) \cdot (\ibox(\Delta_X) \times A)}_{= \Delta_{\ibox X} \times A})
  \\
  & = &
  \Tr_{A,X}^X(\langle f \cdot (\Delta_{\ibox X} \times A), f \cdot (\Delta_{\ibox X} \times A))
  \\
  & = & \sol{(f \cdot (\Delta_{\ibox X} \times A))} & by def.~of $\dagger$.
\end{array}
\]

\subsection{From dagger to trace}

%The axioms of ordinary fixpoint operators in Hasegawa~\cite{h97} contain, in lieu of the double dagger identity, a law called the \emph{Beki\v{c} identity}. We first derive a guarded version of this identity. 

We prove that the operation $\Tr$ defined in Construction~\ref{con:inv}.2 satisfies all the axioms of a guarded trace operator. Again we start with a technical lemma. 

\begin{lemma}
  \label{lem:hTr}
  Let $f: \ibox X \times A \to X \times B$ and define
  \[
  h = (\xymatrix@1@C+1pc{
    \ibox (X \times B) \times A 
    \ar[r]^-{\ibox \prl \times A}
    &
    \ibox X \times A
    \ar[r]^-f
    &
    X\times B
  }).
  \]
  Then we have
  \[
  \Tr_{A,B}^X(f) = (\xymatrix@1{
    A \ar[r]^-{\sol h} & X \times B \ar[r]^-\prr & B
    }).
  \]
\end{lemma}
\begin{proof}
  Notice first that by the simplified composition identity we have
  $\prl \cdot \sol h = \sol{(\prl \cdot f)}$. This implies that
  \[
  \xymatrix@C+1.5pc{
    &
    A \ar[r]^-{\sol h}
    \ar[d]_{\langle \sol h, A\rangle}
    \ar `l[ld] [ldd]_(.4){\langle \sol{(\prl \cdot f)}, A\rangle}
    &
    X \times B
    \\
    &
    X \times B \times A
    \ar[r]_-{p_{X\times B} \times A}
    \ar[ld]_{\prl \times A}
    \ar@{}[l]_-\commu
    \ar@{}[ru]|{\commu\ \text{by ($\dagger$)}}
    \ar@{}[rd]|{\commu\ \text{by nat. of $p$}}
    &
    \ibox (X \times B)  \times A
    \ar[u]_{h}
    \ar[rd]^{\ibox \prl \times A}
    \ar@{}[r]^-\commu
    &
    \\
    X \times A
    \ar[rrr]_-{p_X \times A}
    &&&
    \ibox X \times A
    \ar `u[luu]_(.6)f [luu]
    \
    }
  \]
  The result follows by postcomposing with $\prr$; by the definition
  of $\Tr$ we have
  \[
  \Tr_{A,B}^X (f) = \prr \cdot f \cdot (p_X \times A) \cdot
  \langle \sol{(\prl\cdot f)}, A\rangle = \prr \cdot \sol h.
  \]
\end{proof}
We now verify the properties of a guarded trace for $\Tr$. 

(1)~Vanishing (V1). For any $f: \ibox 1 \times A \to B$ the definition
of $\Tr_{A,B}^1(f)$ yields $f \cdot (p_1 \times A)$; for if
we consider $B$  as the product $1 \times B$ we see that both
$\prl \cdot f: \ibox 1 \times A \to 1$ and its dagger $\sol{(\prl
  \cdot f)} : A \to 1$ are unique morphisms, which implies that $\langle
\sol{(\prl \cdot f)}, A\rangle: A \to 1 \times A$ is the canonical
isomorphism $A \cong 1 \times A$, and $\prr: 1 \times B \to B$ is the canonical isomorphism $1 \times B \cong B$.

(2)~Vanishing (V2). Given $f: \ibox X \times \ibox Y \times A \to X
\times Y \times B$ we form $F = \prl \cdot f: \ibox X \times \ibox Y
\times A \to X$ and $G = \pi_m \cdot f: \ibox X \times \ibox Y \times
A \to Y$, where $\pi_m: X \times Y \times B \to B$ denotes the middle product projection. Then by the Beki\v{c} identity (see Proposition~\ref{prop:bekic}) we
have $\sol{(\langle F, G\rangle \cdot (\can \times A))} = \langle \sol e_L,
\sol e_R\rangle$ for appropriate $e_L: \ibox X \times A \to X$ and
$e_R: \ibox Y \times A \to Y$. From the following diagram we see that
$\prl \cdot \Tr_{\ibox Y \times A, Y\times B}^X (f) = e_R$:
\[
\xymatrix{
  \ibox Y \times A
  \ar[rr]^-{\langle \sol F, \ibox Y \times A\rangle}
  &&
  X \times \ibox Y \times A
  \ar[rr]^-{p_X \times Y \times A}
  &&
  \ibox X \times \ibox Y \times A
  \ar[r]^-f
  \ar[rrd]_G
  \ar@{}[ld]|\commu
  &
  X \times Y \times B
  \ar[r]^-{\prr}
  \ar[rd]^{\pi_m}
  \ar@{}[d]|(.3)\commu
  &
  Y\times B
  \ar[d]^{\prl}
  \ar@{<-} `u[l] `[llllll]_{\Tr^X(f)}^\commu [llllll]
  \ar@{}[ld]^(.25)\commu
  \\
  &&&&&&
  Y
  \ar@{<-} `l[llllllu]^{e_R} [llllllu]
}
\]
By the naturality of $p$ we have
\begin{equation}\label{eq:canp}
  p_X \times p_Y = (\xymatrix@1@C+1pc{
    X \times Y \ar[r]^-{p_{X\times Y}} & \ibox (X\times Y) \ar[r]^-{\can} & \ibox X \times \ibox Y
  }).
\end{equation}
Now we obtain
\[
\begin{array}{rcl}
  \Tr_{A,B}^Y(\Tr_{\ibox Y \times A, Y \times B}^X(f)) 
  & = & 
  \prr \cdot \Tr^X(f) \cdot (p_Y \times A) \cdot \langle \sol e_R, A\rangle \\
  & = & \prr \cdot f \cdot (p_X \times \ibox Y \times A) \cdot 
  \underbrace{\langle \sol F, \ibox Y \times A\rangle \cdot (p_Y
    \times A) \cdot \langle \sol
    e_R, A\rangle}
  \\
  & = & \prr \cdot f \cdot (p_X \times \ibox Y \times A) \cdot 
  \langle \sol F \cdot \langle p_Y \cdot \sol e_R, A\rangle, p_Y \cdot
  \sol e_R, A\rangle
  \\
  & \stackrel{\refeq{eq:eL}}{=} &
  \prr \cdot f \cdot (p_X \times \ibox Y \times A) \cdot \langle \sol e_L, p_Y \cdot \sol e_R, A\rangle
  \\
  & = &
  \prr \cdot f \cdot (p_X \times p_Y \times A) \cdot \langle \sol e_L, \sol e_R, A\rangle
  \\
  & \stackrel{\refeq{eq:canp}}{=} &
  \prr \cdot f \cdot (\can \times A) \cdot (p_{X\times Y} \times A) \cdot \langle \sol e_L, \sol e_R, A\rangle
\end{array}
\]
That this is $\Tr_{A,B}^{X\times Y}(f \cdot (\can \times A))$ now follows from the definition of $\Tr$, the fact that $\langle \sol e_L, \sol e_R\rangle = \sol{(\langle F, G\rangle \cdot (\can \times A))}$ holds by the Beki\v{c} identity and since $\langle F, G \rangle = \prl' \cdot f$ where $\prl': X \times Y \times A \to X \times Y$. 

(3)~Superposing. Let $f: \ibox X \times A \to X \times B$ and denote by $\prr': X \times B \times C \to B \times C$ and $\prl': X\times B \times C \to X$ the projections. Notice first that we have
\[
\begin{array}{rcl}
  \sol{(\prl' \cdot (f \times C))} & = & 
  \big(\xymatrix@1{
    \ibox X \times A \times C 
    \ar[rr]^-{\ibox X \times \prl}
    &&
    \ibox X \times A \ar[r]^-f 
    & 
    X \times B
    \ar[r]^-{\prl}
    &
    X
    }\big)^\dagger
  \\
  & \stackrel{(P)}{=} &
  \big(\xymatrix@1{
    A \times C
    \ar[r]^-\prl
    &
    A
    \ar[rr]^-{\sol{(\prl \cdot f)}}
    &&
    X}\big).
\end{array}
\]
Using this we obtain
\[
\begin{array}{rcl}
  \Tr_{A\times C, B\times C}^X (f \times C) 
  & \stackrel{\text{def}}{=} &
  \prr' \cdot (f \times C) \cdot (p_X \times A \times C) \cdot \underbrace{\langle \sol{(\prl' \cdot (f\times C))}, A\times C\rangle}_{= \langle \sol{(\prl \cdot f)}, A \rangle \times C}
  \\
  & = &
  \big(\prr \cdot f \cdot (p_X \times A) \cdot \langle \sol{(\prl \cdot f)}, A \rangle\big) \times C
  \\
  & \stackrel{\text{def}}{=} & \Tr_{A,B}^X(f) \times C.
\end{array}
\]

(4)~Yanking. Consider $c: \ibox X \times X \to X \times \ibox X$. Then by definition we have
\[
\Tr_{X,\ibox X}^X(c) = \prr \cdot c \cdot (p_X \times X) \cdot \langle \sol{(\prl \cdot c)}, X\rangle.
\]
Thus, we are done if we show that $\sol{(\prl \cdot c)}$ is the identity on $X$, which easily follows from the fixpoint identity:
\[
\sol{(\prl \cdot c)} 
\stackrel{(\dagger)}{=} 
\prl \cdot c \cdot (p_X \times X) \cdot \langle\sol{(\prl \cdot c)}, X\rangle
=
\prr \cdot (p_X \times X) \cdot \langle\sol{(\prl \cdot c)}, X\rangle = X.
\]

(5) Left tightening. Let $f: \ibox X \times A \to X \times B$ and $g: A' \to A$. By the parameter identity we have
\begin{equation}\label{eq:p}
  \sol{(\prl \cdot f \cdot (\ibox X \times g))} = \sol{(\prl \cdot f)} \cdot g.
\end{equation}
Then we have
\[
\begin{array}{rcl}
  \Tr_{A',B}^X(f \cdot (\ibox X \times g)) & \stackrel{\text{def}}{=} & 
  \prr \cdot f \cdot (\ibox X \times g) \cdot (p_X \times A') \cdot \langle \sol{(\prl \cdot f \cdot (\ibox X \times g))}, A'\rangle 
  \\
  & \stackrel{\refeq{eq:p}}{=} & 
   \prr \cdot f \cdot (\ibox X \times g) \cdot (p_X \times A') \cdot \langle \sol{(\prl \cdot f)} \cdot g, A'\rangle
  \\
  & = &
  \prr \cdot f \cdot (p_X \times A) \cdot \langle \sol{(\prl \cdot f)} , A\rangle \cdot g
  \\
  & \stackrel{\text{def}}{=} & 
  \Tr_{A,B}^X(f) \cdot g.
\end{array}
\]

(6)~Right tightening. Let $f: \ibox X \times A \to X \times B$ and $g: B \to B'$. We compute
\[
\begin{array}{rcl}
  \Tr_{A,B'}^X((X \times g) \cdot f) & \stackrel{\text{def}}{=} &
  \underbrace{\prr \cdot (X \times g)}_{= g \cdot \prr} \cdot f \cdot (p_X \times A) \cdot \sol{\langle (\underbrace{\prl \cdot (X \times g)}_{= \prl} \cdot f)}, A\rangle
  \\
  & \stackrel{\text{def}}{=} &
  g \cdot \Tr_{A,B}^X(f). 
\end{array}
\]

(7)~Sliding. Let $f: \ibox X \times A \to X' \times B$ and $g: X' \to X$. Notice first that we have
\begin{equation}\label{eq:sl}
  \sol{(\prl \cdot (g \times B) \cdot f)} = \sol{(g \cdot (\prl \cdot f))} 
  \stackrel{(C)}{=} g \cdot \sol{(\prl \cdot f \cdot (\ibox g \times A))}.
\end{equation}
Then we obtain
\[
\begin{array}{rcl}
  \Tr_{A,B}^X((g \times B) \cdot f) & \stackrel{\text{def}}{=} &
  \underbrace{\prr \cdot (g \times B)}_{= \prr} \cdot f \cdot (p_X \times A) \cdot 
  \langle \sol{(\prl \cdot (g \times B) \cdot f)}, A\rangle
  \\
  & \stackrel{\refeq{eq:sl}}{=} &
  \prr \cdot f \cdot (p_X \times A) \cdot \langle g \cdot \sol{(\prl \cdot f \cdot (\ibox g \times A))}, A \rangle
  \\
  & = &
  \prr \cdot f \cdot \underbrace{(p_X \times A) \cdot (g \times A)}_{= (\ibox g \times A) \cdot (p_{X'} \times A)} \cdot\langle \sol{(\prl \cdot f \cdot (\ibox g \times A))}, A \rangle 
  \\
  & \stackrel{\text{def}}{=} & 
  \Tr_{A,B}^{X'}(f \cdot (\ibox g \times A)).
\end{array}
\]

\subsection{Mutual inverses}

\paragraph{From dagger to trace and back.} We show that $\dagger_{\Tr_\dagger} = \dagger$. For any $f: \ibox X \times A \to X$ we compute
\[
\dagger_{\Tr_\dagger}(f) 
\stackrel{\text{def}}{=} 
(\Tr_\dagger)_{A,X}^X (\langle f, f \rangle)
\stackrel{\text{def}}{=} 
\underbrace{\prr \cdot \langle f,f\rangle}_{= f} \cdot (p_X \times A) \cdot \underbrace{\langle \sol{(\prl \cdot \langle f, f\rangle)}, A\rangle}_{= \langle \sol f, A\rangle}
\stackrel{(\dagger)}{=} 
\sol f.
\]

\paragraph{From trace to dagger and back.} We show that $\Tr_{\dagger_\Tr} = \Tr$. Let $f: \ibox X \times A \to X \times B$. Then we have by the definition of $\dagger_\Tr$
\begin{equation}\label{eq:aux}
  (\prl \cdot f)^{\dagger_\Tr} 
  \stackrel{\text{def}}{=}
  \Tr_{A,X}^X (\langle \prl \cdot f, \prl \cdot f\rangle) 
  = 
  \Tr_{A,X}^X (\Delta_X \cdot \prl \cdot f) 
  =: 
  h.
\end{equation}
Using Lemma~\ref{lem:fptr}, this allows us to conclude
\[
(\Tr_{\dagger_\Tr})_{A,B}^X(f) 
\stackrel{\text{def}}{=} 
\prr \cdot f \cdot (p_X \times A) \cdot \underbrace{\langle (\prl \cdot f)^{\dagger_\Tr}, A\rangle}_{= \langle h, A\rangle}
\stackrel{\text{Lem.~\ref{lem:fptr}}}{=}
\Tr_{A,B}^X (f).
\]

\section{Proof of Theorem~\ref{thm:unif}}
\label{app:unif}

\paragraph{1.~From trace to dagger.}
Let $f$, $f'$ and $h$ form the commutative square on the left below:
\[
\xymatrix{
  \ibox X \times A 
  \ar[r]^-f
  \ar[d]_{\ibox h \times A}
  \ar@{}[rd]|\commu
  &
  X
  \ar[d]^h
  \\
  \ibox X' \times A 
  \ar[r]_-{f'}
  &
  X'
}
\qquad\qquad
\xymatrix{
  \ibox X \times A
  \ar[r]^-{\langle f,f\rangle}
  \ar[d]_{\ibox h \times A}
  &
  X \times X
  \ar[rd]_{h \times h}
  \ar[r]^-{X \times h}
  &
  X \times X'
  \ar[d]^{h \times X'}
  \ar@{}[ld]|(.35)\commu
  \\
  \ibox X' \times A
  \ar[rr]_-{\langle f',f'\rangle}
  \ar@{}[ru]_(.65)\commu
  &&
  X'\times X' 
}
\]
Then the diagram on the right above commutes, too, and thus, by uniformity of $\Tr$ we have
\begin{equation}\label{eq:unifTr}
  \Tr_{A,X'}^X((X\times h) \cdot \langle f, f\rangle) = \Tr_{A, X'}^{X'}\langle f', f'\rangle.
\end{equation}
Thus, we obtain: 
$\ 
h \cdot \sol f 
\stackrel{\text{def}}{=} 
h \cdot \Tr_{A,X}^X\langle f, f \rangle
\stackrel{(Rt)}{=} 
\Tr_{A,X'}^{X} ((X \times h) \cdot \langle f, f\rangle)
\stackrel{\refeq{eq:unifTr}}{=}
\Tr_{A,X'}^{X'}\langle f', f' \rangle
\stackrel{\text{def}}{=} 
\sol{(f')}.$

\paragraph{2.~From dagger to trace.}
Let $f$, $f'$ and $h$ form the commutative square on the left below:
\[
\xymatrix{
  \ibox X \times A 
  \ar[r]^-f
  \ar[d]_{\ibox h \times A}
  \ar@{}[rd]|\commu
  &
  X \times B
  \ar[d]^{h \times B}
  \\
  \ibox X' \times A 
  \ar[r]_-{f'}
  &
  X' \times B
}
\qquad\qquad
\xymatrix@C+1pc{
  \ibox (X \times B) \times A
  \ar[r]^-{\ibox\prl \times A}
  \ar[d]_{\ibox (h \times B) \times A}
  &
  \ibox X \times A
  \ar[r]^-f
  \ar[d]_{\ibox h \times A}
  \ar@{}[ld]|\commu
  \ar@{}[rd]|\commu
  &
  X\times B
  \ar[d]^{h \times B}
  \\
  \ibox (X' \times B) \times A
  \ar[r]_-{\ibox \prl \times A}
  &
  \ibox X' \times A
  \ar[r]_-{f'}
  &
  X' \times B
}
\]
Then the diagram on the right above commutes, too, and thus, by uniformity of $\dagger$ we have
\begin{equation}\label{eq:unifdag}
  (h \times B) \cdot \sol{(f \cdot (\ibox \prl \times A))} 
  = 
  \sol{(f' \cdot (\ibox\prl \times A))}
\end{equation}
Using Lemma~\ref{lem:hTr} we now compute:
\[
\begin{array}{rcl}
\Tr_{A,B}^{X'}(f') 
& \stackrel{\text{Lem.~\ref{lem:hTr}}}{=} &
\prr \cdot \sol{(f' \cdot (\ibox \prl \times A))} \\
& \stackrel{\refeq{eq:unifdag}}{=} &
\underbrace{\prr \cdot (h \times B)}_{= \prr} \cdot \sol{(f \cdot (\ibox \prl \times A))}\\
%& = &
%\prr \cdot \sol{(f \cdot (\ibox \prl \times A))} \\
& \stackrel{\text{Lem.~\ref{lem:hTr}}}{=} &
\Tr_{A,B}^X(f).
\end{array}
\]

\section{Proof of Corollary~\ref{cor:iso}}
\label{app:iso}

1.~Let $F: (\C, \ibox^\C, \Tr_{\C}) \to (\D, \ibox^\D, \Tr_{\D})$ be a morphism of guarded traced categories. We show that $F$ preserves $\dagger_\Tr$ as defined in Construction~\ref{con:inv}.1. Let $f: \ibox X \times A \to X$ in $\C$. Then we have (dropping subscripts of $\dagger$ and $\Tr$) 
\[
\begin{array}{rclp{5cm}}
  F(\sol f) & = & F(\Tr^X\langle f, f\rangle) & (by definition of $\dagger$) \\
  & = & \Tr^{FX}(F\langle f, f\rangle) & ($F$ trace preserving) \\
  & = & \Tr^{FX}\langle Ff, Ff\rangle & ($F$ finite product preserving) \\
  & = & \sol{(F f)} & (by definition of $\dagger$).
\end{array}
\]

2.~Let $F: (\C, \ibox^\C, \dagger) \to (\D, \ibox^\D, \ddagger)$ be a morphism of guarded Conway categories. We show that $F$ preserves $\Tr_\dagger$ as defined in Construction~\ref{con:inv}.2. Let $f:\ibox^\C X \times A \to X \times B$ in $\C$. Then we have (again we drop all subscripts of $\Tr$ and $\dagger$)
\[
\begin{array}{rclp{5cm}}
  F(\Tr^X(f)) & = & F\left(\prr \cdot (f \cdot (\ibox^\C \prl\times A))^\dagger\right) & (by Lemma~\ref{lem:hTr}) \\
  & = & F \prr \cdot \left(Ff \cdot F(\ibox^\C \prl \times A)\right)^\ddagger & ($F$ dagger preserving) \\
  & = & \prr \cdot (Ff \cdot (F(\ibox^\C \prl) \times FA))^\ddagger & ($F$ finite product preserving) \\
  & = & \prr \cdot (Ff \cdot (\ibox^\D \underbrace{F\prl}_{=\prl} \times FA))^\ddagger & (by \refeq{eq:sat}) \\
  & = & \Tr^{FX}(Ff) & (by Lemma~\ref{lem:hTr}).
\end{array}
\]
This completes the proof.
\end{document}